%% file: new.tex
\documentclass{irmaart}
\input{ha.sty}

\usepackage[latin1]{inputenc}
\usepackage[T1]{fontenc}

\usepackage{amsmath}
\usepackage{multirow}
\usepackage{graphics}
\usepackage[dvips]{epsfig}
\usepackage{algorithm,algorithmic}
\usepackage{subfigure}

\usepackage{float}%for floating grammars
\floatstyle{plain}
\newfloat{grammar}{htbp}{gra}[section]
\floatname{grammar}{\bf Grammar}

\def \Que{{\mathbb Q}}

\def \sc{{\rm sc}}

\DeclareMathOperator{\ord}{ord}

\def\bull{\bullet} 
\def\orr{{\ | \ }}

\def\rpn{{\rm rpn}}
\def\alphab{{\rm alph}}

\def\Zee{{\mathbb Z}}
\def\derives{\Longrightarrow}
\def\derivestar{\Longrightarrow^*}

\renewcommand{\epsilon}{\varepsilon}

\begin{document}

%====== Full title =====================================
\title{Enumerating regular expressions and their languages}

%====== Full authors ===================================
\author{Hermann Gruber$^1$, Jonathan Lee$^2$, Jeffrey Shallit$^3$}

%====== Running head: short form for authors and short tilte
\markboth{H.~Gruber, J.~Lee, J.~Shallit}{Enumerating regular expressions}

%====== Address and emails ============================
\address{
$^{1}$ Institut f\"ur Informatik, Justus-Liebig-Universit\"at Giessen\\
Arndtstrasse 2\\
D-35392 Giessen, Germany\\[2mm]
$^{2}$ Department of Mathematics, Stanford University\\
Building 380, Sloan Hall\\
Stanford, CA 94305, United States of America\\[2mm]
$^{3}$ School of Computer Science, University of Waterloo\\
Waterloo, ON N2L 3G1, Canada\\
email:\,\url{hermann.gruber@informatik.uni-giessen.de},\\ 
\url{jlee@math.stanford.edu}, \\
\url{shallit@cs.uwaterloo.ca}
}

%====== Makes the title and labels the chapter ========
\maketitle\label{chapterERL} 
%=====================================================

%======= Mandatory classification and keywords ========


%======= produces a table of contents of the chapter
\localtableofcontents

%======================================================
%
% Your contribution starts here
%
%======================================================

\section{Introduction and overview}

Regular expressions have been studied for almost fifty years, yet
many interesting and challenging problems about them remain unsolved.
By a regular expression, we mean a string over the alphabet
$ \Sigma \ \cup \{\, +, *,\left(\right.,\left.\right), \epsilon, \emptyset \,\}$
that represents a regular language.
For example, ${\tt (0+10)*(1+ \epsilon)}$ represents the
language of all strings over $\lbrace \texttt{0,1} \rbrace$ that
do not contain two consecutive \texttt{1}'s.

      We would like to enumerate both (i) valid regular expressions and
(ii) the
distinct languages they represent. Observe that these are two
different enumeration tasks: on the one hand, every regular expression
represents exactly one regular language. On the other hand, 
simple examples, 
such as the expressions $\tt (a+b)*$ and $\tt (b*a*)*$, show that 
there is no one-to-one correspondence 
between regular languages and regular expressions. 

We are in a similar situation if we use descriptors other than
regular expressions, such as deterministic or nondeterministic 
finite automata.    Although enumeration of automata has a long
history, until recently little 
attention was paid to enumerating the distinct languages accepted.   Instead
authors concentrated on enumerating the automata themselves according
to various criteria (e.g., acyclic, nonisomorphic, strongly
connected, initially connected, ...).  

Here is a brief survey of known results on automata.  Vyssotsky
\cite{Vyssotsky:1959} raised the question of enumerating strongly
connected finite automata in an obscure technical report (but we have
not been able to obtain a copy).  Harary \cite{Harary:1959} enumerated
the number of ``functional digraphs'' (which are essentially unary
deterministic automata with no distinguished initial or final states)
according to their cycle structure; also see Read \cite{Read:1961} and
\cite{Livshits:1964}.
Harary also mentioned the
problem of enumerating deterministic
finite automata over a binary alphabet as an open problem in a
1960 survey of open problems in enumeration \cite[pp.\ 75,87]{Harary:1960},
and later in a similar 1964 survey \cite{Harary:1964}.
Ginsburg \cite[p.\ 18]{Ginsburg:1962} asked for the number of nonisomorphic
automata with output on $n$ states with given input and output alphabet size.

Harrison \cite{Harrison:1964,Harrison:1965} developed exact formulas
for the number of automata with specified size of the input alphabet,
output alphabet, and number of states.  
Similar results were found
by Korshunov \cite{Korshunov:1966}.  However, in their model, the
automata do not have a distinguished initial state or set of final
states.  Using the same model, Radke \cite{Radke:1965} enumerated the
number of strongly connected automata, but his solution was very
complicated and not particularly useful.  
Harary and Palmer \cite{Harary&Palmer:1967} found very complicated
formulas in the same model, but including an initial state and any
number of final states.

Harrison \cite{Harrison:1964,Harrison:1965} gave asymptotic estimates
for the number of automata in his model, but his formulas contained
some errors that were later corrected by Korshunov \cite{Korshunov:1967}.
For example, the number of nonisomorphic unary automata with $n$ states
(and no distinguished initial or final states) is asymptotically $c
(\pi n)^{-{1 \over 2}} \tau^{-n}$ where $c \doteq 0.80$ and $\tau
\doteq 0.34$.

Much work on enumeration of automata was done in the former Soviet 
Union.  For example,
Liskovets \cite{Liskovets:1969} studied the number of initially connected
automata and gave both a recurrence formula and an asymptotic formula for
them; also see Robinson \cite{Robinson:1985}.
Korshunov \cite{Korshunov:1974} counted the number of minimal
automata, and \cite{Korshunov:1975} gave asymptotic estimates for the
number of initially connected automata.  The 78-page survey by Korshunov
\cite{Korshunov:1978}, which unfortunately seems to never have been
translated into English, gives these and many other results.
More recently, Bassino and Nicaud \cite{Bassino&Nicaud:2007} found that
the number of nonisomorphic initially connected deterministic automata
with $n$ states is closely related to the Stirling numbers of the
second kind.

Shallit and Breitbart observed that
the number of finite automata can be applied to give bounds on the
``automaticity'' of languages and functions \cite{Shallit&Breitbart:1996}.
Pomerance, Robson, and Shallit \cite{Pomerance&Robson&Shallit:1997}
gave an upper bound on the number of distinct unary languages accepted
by unary NFA's with $n$ states.
Domaratzki, Kisman, and Shallit considered the number of distinct
languages accepted by finite automata with $n$ states
\cite{Domaratzki&Kisman&Shallit:2002}.  They showed, for example, that
the number of distinct languages accepted by unary finite
automata with $n$ states is $2^n (n - \alpha + O(n 2^{-n/2}))$, where
$\alpha \doteq 1.3827$.  (A weaker result was previously obtained
by Nicaud \cite{Nicaud:1999}.)  
Domaratzki \cite{Domaratzki:2004,Domaratzki:2004b} gave bounds on the
number of minimal DFA's accepting finite languages, which were improved by
Liskovets \cite{Liskovets:2006}.  Also see \cite{Callan:2008}.
For more details about enumeration of automata and languages, see
the survey of Domaratzki \cite{Domaratzki:2006}.

\section{On measuring the size of a regular expression}

Although, as we have seen, there has been much work for over 50 years on 
enumerating automata and the languages they represent, 
the analogous problem for regular expressions does not
seem to have been studied before 2004 \cite{Lee&Shallit:2005}.
We define $R_k(n)$ to be the number of distinct languages
specified by regular expressions of size $n$ over a $k$-letter alphabet.
The ``size'' of a regular expression can be defined in several
different ways~\cite{Ellul&Krawetz&Shallit&Wang:2005}:

\begin{itemize}

\item \textsl{Ordinary length}:  total number of symbols, including
parentheses, $\emptyset$, $\epsilon$, etc., counted with multiplicity.

	\begin{itemize}

	\item $\tt{(0+10)*(1+\epsilon)}$ has ordinary length 12

	\item Mentioned, for example,
	in \cite[p.\ 396]{Aho&Hopcroft&Ullman:1974},
	\cite{Ilie&Yu:2002}.

	\end{itemize}

\item \textsl{Reverse polish length}:  number of symbols in a reverse polish
equivalent, including a symbol $\bull$ for concatenation.
Equivalently, number of nodes in a syntax tree for the expression.

	\begin{itemize}

	\item ${\tt(0+10)*(1+\epsilon)}$ in reverse polish 
	would be ${\tt 010\bull+*\epsilon+\bull}$

	\item This has reverse polish length $10$

	\item Mentioned in \cite{Ziadi:1996}

	\end{itemize}

\item \textsl{Alphabetic width}:  number of symbols from $\Sigma$, counted
with multiplicity, not including
$\epsilon$, $\emptyset$, parentheses, operators

	\begin{itemize}

	\item ${\tt(0+10)*(1+\epsilon)}$ has alphabetic width $4$

	\item Mentioned in \cite{McNaughton&Yamada:1960,Ehrenfeucht&Zeiger:1976,Leiss:1980b}

	\end{itemize}
\end{itemize}

Each size measure seems to have its own advantages and disadvantages. 
The ordinary length appears to be the most direct way to measure the size 
of a regular expression.  Here we can employ the usual 
priority rules, borrowed from arithmetic,
for saving parentheses and omitting the $\bull$ 
operator.
This favors the catenation operator~$\bull$ over 
the union operator~$+$.
For instance, the expression ${\tt(a\bull b)+(c \bull d)}$ 
can be written more briefly as ${\tt ab+cd}$, which has ordinary 
length~$5$, whereas there 
is no corresponding way to simplify 
the expression ${\tt (a+b)(c + d)}$, which is twice as long.
The other two measures are more robust in this respect. 
In 
particular, reverse polish length is a faithful measure for 
the amount of memory required to store the parse tree of a 
regular expression, and alphabetic width is often used 
in proofs of upper and lower bounds, 
compare~\cite{Holzer&Kutrib:2010}. 
A drawback of alphabetic width 
is that it may be far from the ``real'' size of a given 
regular expression. As an example, the expression 
${\tt ((\epsilon + \emptyset)*\emptyset + \epsilon)*}$ 
has alphabetic width $0$. 

    However, these three measures
are all essentially identical, up to a constant multiplicative factor.
We say ``essentially'' because one can always artificially inflate the 
ordinary length of a regular expression by adding arbitrarily many 
multiplicative factors of $\epsilon$, additive factors of $\emptyset$, etc.
In order to avoid such trivialities, we define what it means for
a regular expression to be collapsible, as follows:
\begin{definition}
Let $E$ be a regular expression over the alphabet $\Sigma$, and
let $L(E)$ be the language specified by $E$.  We say
$E$ is {\sl collapsible} if any of the following conditions hold:
\begin{enumerate}
\item $E$ contains the symbol $\emptyset$ and $|E| > 1$;
\item $E$ contains a subexpression of the form $FG$ or $GF$ where $L(F) = 
\lbrace \epsilon \rbrace$;
\item $E$ contains a subexpression of the form $F+G$ or $G+F$ 
where $L(F) = \lbrace \epsilon \rbrace$ and $\epsilon \in L(G)$.
\end{enumerate}
Otherwise,
if none of the conditions hold, $E$ is said to be {\sl uncollapsible}.
\end{definition}

\begin{definition}
If $E$ is an uncollapsible regular expression such that
\begin{enumerate}
\item $E$ contains no superfluous parentheses; and
\item $E$ contains no subexpression of the form ${F^*}^*$.
\end{enumerate}
then we say $E$ is {\sl irreducible}.
\end{definition}

    Note that a minimal regular expression for $E$ is uncollapsible and 
irreducible, but the converse does not necessarily hold.
     In \cite{Ellul&Krawetz&Shallit&Wang:2005} the following
theorem is proved (cf.\ \cite{Ilie&Yu:2002}).

\begin{theorem}
Let $E$ be a regular expression over $\Sigma$.
Let $|E|$ denote its ordinary
length, let $|\rpn(E)|$ denote its reverse polish length, and let
$|\alphab(E)|$ denote the number of alphabetic symbols contained in $E$.
Then we have
\begin{enumerate}
\item[(a)] $|\alphab(E)|\leq |E|$;
\item[(b)] If $E$ is irreducible and $|\alphab(E)| \geq 1$,
then $|E| \leq 11 \cdot |\alphab(E)| -4$;
\item[(c)] $|\rpn(E)| \leq 2 \cdot |E|-1$;
\item[(d)] $|E| \leq 2 \cdot |\rpn(E)| - 1$;
\item[(e)] $|\alphab(E)| \leq {1 \over 2}(|\rpn(E)|+1)$;
\item[(f)] If $E$ is irreducible and $|\alphab(E)| \geq 1$, then
$|\rpn(E)| \leq 7 \cdot |\alphab(E)| -2$.
\end{enumerate}
\label{irred}
\end{theorem}

\section{A simple grammar for valid regular expressions}

%The imprecise notion of what constitutes a regular expression.
%Making this more precise by defining
%valid regular expressions with a grammar.  A simple grammar for valid
%regular expressions.

As we have seen,
if we want to enumerate regular expressions by size,
we first have to agree upon a notion of expression 
size. But even then there still remains some ambiguity about the definition
of a valid regular 
expression.  For example, does the empty expression, that is, a string of length
zero, constitute a valid regular expression? How about
\texttt{()} or \texttt{a**}?
The first two, for example, generate errors in the software package 
Grail version 2.5 \cite{Raymond&Wood:1994}.  
      Surprisingly, very few textbooks, if any, define valid regular
expressions properly or formally.  For example, using the definition
given in Martin \cite[p.\ 86]{Martin:2003}, the expression $\tt 00$ is not
valid, since it is not fully parenthesized.  (To be fair, after the
definition it is implied that parentheses can be omitted in some cases,
but no formal definition of when this can be done is given.)
      Probably the best way to define valid regular expressions is
with a grammar.  We now present an unambiguous grammar for all
valid regular expressions:

\begin{eqnarray*}
S & \rightarrow & E_{+} \orr E_{\bull} \orr G \\
E_{+} & \rightarrow & E_{+} + F \orr F+F \\
F & \rightarrow & E_{\bull} \orr G \\
E_{\bull} & \rightarrow &  E_{\bull} \, G \orr GG \\
G & \rightarrow & E_{*} \orr C \orr P \\
C & \rightarrow & \emptyset \orr \epsilon \orr  a \ \ \ (a \in \Sigma) \\
E_{*} & \rightarrow & G* \\
P & \rightarrow & (S)
\end{eqnarray*}

This grammar can be proved unambiguous by induction on the size of the
regular expression generated.
The meaning of the variables is as follows:

\begin{itemize}

\item[$S$] generates all regular expressions

\item[$E_{+}$] generates all unparenthesized expressions where the last
               operator was $+$

\item[$E_{\bull}$] generates all unparenthesized expressions where the last
                operator was $\cdot$ (implicit concatenation)

\item[$E_{*}$] generates all unparenthesized expressions where the last
                operator was $*$ (Kleene closure)

\item[$C$] generates all unparenthesized expressions where there was
                no last operator (i.e., the constants)

\item[$P$] generates all parenthesized expressions

\end{itemize}

Here by ``parenthesized" we mean there is at least one pair of enclosing
parentheses.  Note this grammar allows $\tt{a**}$, but disallows $\tt{()}$.
Once we have an unambiguous grammar, we can use a powerful 
tool --- the Chomsky-Sch\"utzenberger theorem --- to enumerate
the number of expressions of size~$n$.

\section{Unambiguous context-free grammars and the Chomsky-Sch\"utzenberger theorem}

     Our principal tool for enumerating the number of strings of length
$n$ generated by an unambiguous context-free grammar is the
Chomsky-Sch\"utzenberger theorem \cite{Chomsky&Schutzenberger:1963}.  To
state the theorem, we first recall some basic notions about grammars;
these can be found in any introductory textbook on formal language
theory, such as \cite{Hopcroft&Ullman:1979}.

     A {\it context-free grammar} is a quadruple of the form
$G = (V, \Sigma, P, S)$, where $V$ is a nonempty finite set of variables,
$\Sigma$ is a nonempty finite set called the {\it alphabet}, $P$ is a
finite subset of $V \times (V \ \cup \ \Sigma)^*$ called the 
{\it productions}, and $S \in V$ is a distinguished variable called the
{\it start variable}.  The elements of $\Sigma$ are often called terminals.
A production $(A, \gamma)$ is typically
written $A \rightarrow \gamma$.  
     A {\it sentential form} is an element of $(V \ \cup \ \Sigma)^*$.
Given a sentential form $\alpha A \beta$, where $A \in V$ and 
$\alpha, \beta \in (V \ \cup \ \Sigma)^*$, we can apply the production
$A \rightarrow \gamma$ to get a new sentential form $\alpha \gamma \beta$.
In this case we write $\alpha A \beta \derives \alpha \gamma \beta$.
We write $\derivestar$ for the reflexive, transitive closure of $\derives$;
that is, we write $\alpha \derivestar \beta$ if we can get from $\alpha$
to $\beta$ by $0$ or more applications of $\derives$.
     The language generated by a context-free grammar is the set of all
strings of terminals obtained in~$0$ or more derivation steps from~$S$, the
start variable.  Formally, $L(G) = \lbrace x \in \Sigma^* \ : \ 
S \derivestar x \rbrace$.  A language is said to be {\it context-free} if it is
generated by some context-free grammar.
Given a sentential form $\alpha$ derivable from a variable $A$, we
can form a {\it parse tree} for $\alpha$ as follows:  the root is
labeled $A$. Every node labeled with a variable $B$
has subtrees with roots labeled, from left to right, with the elements 
of $\gamma$, where $B \rightarrow
\gamma$ is a production.  A grammar is said to be {\it unambiguous} if for
each $x \in L(G)$ there is exactly one parse tree for $x$; otherwise it is
said to be {\it ambiguous}.  It is known that not every context-free
language has
an unambiguous grammar.  

     Now we turn to formal power series; for more
information, see, for example \cite{Wilf:2006}.
A formal power series over a commutative ring
$R$ in an indeterminate $x$ is an infinite sequence of coefficients 
$(a_0, a_1, a_2, \ldots)$ chosen from $R$, and usually written 
$a_0 + a_1 x + a_2 x^2 + \cdots$.  The set of all such formal power series
is denoted $R[[x]]$.  The set of all formal power series is itself a
commutative ring, with addition defined term-by-term,
and multiplication defined by the usual Cauchy product
as follows:  if $f = a_0 + a_1 x + a_2 x^2 + \cdots $ and
$g = b_0 + b_1 x + b_2 x^2 + \cdots$, then
$fg = c_0 + c_1 x + c_2 x^2 + \cdots$, where 
$c_n = \sum_{i+j = n} a_i b_j$.    Exponentiation of formal series
is defined, as usual, by iterated multiplication, so that $f^2 = ff$, for
example.
A formal power series $f$ is said to be {\it algebraic} (over $R(x)$) if there
exist a finite number of polynomials with coefficients in $R$,
$r_0 (x), r_1(x), \ldots, r_n (x)$ such that
$$r_0 (x) + r_1(x)f + \cdots + r_n (x) f^n = 0.$$
 
The simplest nontrivial
examples of algebraic formal series are the {\it rational functions}, which
are quotients of polynomials $p(x)/q(x)$.   
Here is a less trivial example.  
The generating function of the Catalan numbers
$$f(x) = \sum_{n \geq 0} {{{2n} \choose n} \over {n+1}} x^{n+1}
= x + x^2 + 2x^3 + 5x^4 + 14x^5 + 42x^6 + 132x^7 + \cdots ,$$
is well known \cite{Stanley:1999} to satisfy 
$f(x) = {1 \over 2}(1 - \sqrt{1-4x})$,
and hence we have $f^2 -f + x = 0$.  Thus $f(x)$ is an algebraic (even quadratic!) formal series.

Now that we have the preliminaries, we can state the Chomsky-Sch\"utzenberger
theorem:

\begin{theorem}
\label{Chomsky-Schutzenberger}
If $L$ is a context-free language
having an unambiguous grammar, and $a_n := | L \ \cap \Sigma^n |$,
then $\sum_{n \geq 0} a_n x^n$ is a formal power
series in $\Zee[[x]]$ that is algebraic over
$\Que(x)$.
\end{theorem}

Furthermore, the equation of which the formal power series
is a root can be deduced as follows:  first, we carry out the following
replacements:
\begin{itemize}

\item Every terminal is replaced by a variable $x$

\item Every occurrence of $\epsilon$ is replaced by the integer $1$

\item Every occurrence of $\rightarrow$ is replaced by $=$

\item Every occurrence of $|$ is replaced by $+$

\end{itemize}
By doing so, we get a system of algebraic equations, called the
``commutative image'' of the grammar, which can then be solved
to find a defining equation for the power series.
Oddly enough, Chomsky and Sch\"utzenberger did not actually provide
a proof of their theorem.  A proof is given by Kuich and
Salomaa~\cite{Kuich&Salomaa:1985} and, more recently, by 
Panholzer~\cite{Panholzer:2005}.

Let's look at a simple example.  Consider the unambiguous grammar 
\begin{eqnarray*}
S & \rightarrow & M \ | \ U \\
M & \rightarrow & 0M1M \ | \ \epsilon \\
U & \rightarrow & 0S \ | \ 0M1U 
\end{eqnarray*}
which represents strings of ``if-then-else'' clauses.
Then this grammar has the following commutative image:
\begin{eqnarray}
S &=& M + U  \label{eq1}\\
M &=& x^2 M^2 + 1  \label{eq2}\\
U &=& Sx + x^2 MU \label{eq3} 
\end{eqnarray}

      This system of equations has the following power series solutions:
\begin{eqnarray*}
M &=& 1 + x^2 + 2x^4 + 5x^6 + 14x^8 + 42x^{10} +  \cdots \\
U &=& x + x^2 + 3x^3 + 4x^4 + 10x^5 + 15x^6 + 35x^7 + 56x^8 + \cdots \\
S &=& 1 + x + 2x^2 + 3x^3 + 6x^4 + 10x^5 + 20x^6 + 35x^7 + \cdots
\end{eqnarray*}

By the Chomsky-Sch\"utzenberger theorem, each variable satisfies an
algebraic equation over $\Que(x)$.  We can solve the system above to find
the equation for $S$, as follows:
first, we solve (\ref{eq3}) to get $U = {{Sx} \over {1-x^2 M}}$, and
substitute back in (\ref{eq1}) to get $S = M + {{Sx} \over {1-x^2 M}}$.
Multiplying through by $1-x^2 M$ gives
$S - x^2 MS = M - x^2 M^2 + Sx$, which, by (\ref{eq2}), is
equivalent to $S - x^2MS = 1+Sx$.  Solving for $S$, we get
$S = {1 \over {1-x^2M-x}}$.
Now (whatever $M$ and $x$ are) we have
$$(1-x^2M-x)^2 = x^2(1-M+x^2M^2) - x(2x-1) -(2x-1)(1-x^2M - x),$$
so we get $S^{-2} = -x(2x-1) - (2x-1)S^{-1}$ and hence
$$x(2x-1)S^2 + (2x-1)S + 1 = 0.$$
This is an equation for~$S$.

\section{Solving algebraic equations using Gr\"obner bases}\label{sec:groebner}
%Solving algebraic equations using Gr\"obner bases.  Examples using
%various computer algebra systems.
\input{groebner}

\section{Asymptotic bounds via singularity analysis}\label{sec:analysis}
\input{analysis}

\section{Lower bounds on enumeration of regular languages by regular
  expressions}
\input{lower}

\section{Upper bounds on enumeration of regular languages by regular
  expressions}\label{sec:upper}
\input{upper}

\section{Exact enumerations}
\input{exact}

\section{Conclusion and open problems}
In this chapter, we discussed various approaches to enumerating regular expressions 
and the languages they represent, and we used algebraic and analytic tools to compute 
upper and lower bounds for these enumerations. 
Our upper and lower bounds are not always very close, so an obvious open problem (or 
class of open problems) is to improve these bounds.
Other problems we did not examine here involve enumerating interesting subclasses of 
regular expressions.  For example, in linear expressions, every alphabet symbol occurs 
exactly once.  In addition to the intrinsic interest, enumerating subclasses may 
provide a strategy for improving the lower bounds for the general case.

\bibliographystyle{new2}
\bibliography{abbrevs,ciaa,new,enum,re}
%======================================================
% The abstract on a new page, at the end
%======================================================

\newpage
\abstract{
In this chapter we discuss the problem of enumerating distinct regular
expressions by size and the regular languages they represent.  We discuss 
various notions of the size of a regular expression that appear in the
literature and their advantages and disadvantages.
We consider a formal definition
of regular expressions using a context-free grammar.  

We then show how to enumerate strings generated by an unambiguous
context-free grammar using the Chomsky-Sch\"utzenberger theorem.  This
theorem allows one to construct an algebraic equation whose power
series expansion provides the enumeration.  Classical tools from
complex analysis, such as singularity analysis, can then be used to
determine the asymptotic behavior of the enumeration.

We use these algebraic and analytic methods to obtain asymptotic
estimates on the number of regular expressions of size~$n$. A single
regular language can often be described by several regular expressions,
and we estimate the number of distinct languages denoted by regular
expressions of size~$n$.  We also give asymptotic estimates for these
quantities. For the first few values, we provide exact enumeration
results.
}

\end{document}

%% file: groebner.tex
%\documentclass{article}
%\usepackage{amsmath, amssymb}
% Pierre ;-)
%\newcommand{\burp}[1]{\noindent\textbf{BURP:}\marginpar{****}
%   \textit{{#1}}\marginpar{****}
%   \textbf{:PRUB}}
%\title{Solving Algebraic Equations via Gr\"obner Bases}

%\begin{document}

%\maketitle

Before introducing the notion of Gr\"obner bases, we
describe some of the relevant mathematical notions from the
field of \emph{commutative algebra}. The exposition here is
impressionistic; readers familiar with algebraic geometry will
have no difficulty reformulating it in more formalized terms. 
For readers seeking for a more thorough introduction into the topic, 
there are accessible textbooks at the undergraduate level, such 
as~\cite{Cox&Little&OShea:2007}; a standard graduate level textbook  
is~\cite{Hartshorne:1977}.

We recall that a \emph{field} $k$ is a commutative ring with the
additional property that multiplicative inverses exist.
That is, for any non-zero $a \in k$, there exists an element $b$
such that $ab = ba = 1$; more informally, one can ``divide by $a$''.
Familiar examples of fields are the rational numbers $\mathbb{Q}$,
the real numbers $\mathbb{R}$, and the complex numbers $\mathbb{C}$.
On the other hand, the commutative ring $\mathbb{Z}$ of integers
is not a field, and the smallest field containing it is $\mathbb{Q}$.

For our application to the asymptotic enumeration of regular languages,
we are interested in the commutative ring of formal power series
$\mathbb{Z}[[x]]$.
This is not a field, but rather only a ring --- note, for example, that
the element $2x$ does not have a multiplicative inverse.
For the purposes of our
algebraic framework it is convenient to work with the field
$k = \mathbb{Q}((x))$ of formal Laurent series over $\mathbb{Q}$.
A formal Laurent series is defined similarly to a formal power series,
with the difference that finitely many negative exponents are allowed;
an example is
%\[ \exp(x)/x^2 = 1/x^2 + 1/x + 1/2 + x/6 + x^2/24 + \cdots \,. \]
\[ {e^x\over x^2} = {1 \over {x^2}} + {1\over x} + {1 \over 2} + {x \over 6} + {{x^2} \over {24}} + \cdots \,. \]
The following discussion holds for any field $k$, but for intuition,
the reader may prefer to think of $k = \mathbb{R}$.

Given any field $k$ and indeterminates $X_1, X_2, \ldots, X_n$, there
are two important objects:
\begin{itemize}
\item the $n$-dimensional vector space $W = k^n$ over $k$,
with coordinates $X_i$ ($1 \leq i \leq n$); and
\item the ring $k[X_1, X_2, \ldots, X_n]$ of (multivariate) polynomials 
over~$k$ in $n$ indeterminates.
\end{itemize}
For instance, taking $k=\mathbb{Q}((x))$, 
the polynomial $Sx + x^2MU - U$, which we used 
in the previous section in Equation \eqref{eq3},
is member of the ring $k[S,M,U]$. The corresponding 
vector space $W$ has coordinates~$S$, $M$, and~$U$. 
Notice that~$x$ is not a coordinate 
of~$W$, but an artifact originating 
from the way the members of~$k$ are defined.

Given any collection of polynomials $\mathcal{F}$ in $R$,
we can define their \emph{vanishing set} $V(\mathcal{F})$
to be the set of common solutions in $W$; that is, all points
$(x_1, x_2, \ldots, x_n) \in W$ such that
\[ f(x_1, x_2, \ldots, x_N) = 0 \quad \text{for all } f \in \mathcal{F} \,. \]
As an example, let $W = \mathbb{R}^3$, with coordinates $X,Y,Z$.
Then, the vanishing set of the set of polynomials
$\mathcal{F} = \{X, Y+3, Z+Y-2\}$ is the
single point given by $(X,Y,Z) = (0,-3, 5)$; the vanishing set of
the single polynomial $Z - X^2 - Y^2$ is an upward-opening paraboloid.

The \emph{ideal} $\left< \mathcal{F} \right>$ generated by a
collection $\mathcal{F}$ of polynomials is the set of all
$R$-linear combinations of $\mathcal{F}$; that is, all polynomials
of the form
\[ p_1 \cdot f_1 + p_2 \cdot f_2 + \cdots + p_\ell \cdot f_\ell 
    \quad \text{where } p_i \in R, f_i \in \mathcal{F}
    \text{ for all } i \,. \]
Observe that the vanishing sets of a collection of polynomials
and their generated ideal are equal:
$V(\mathcal{F}) = V(\left< \mathcal{F} \right>)$.

A \emph{term ordering} on $R$ is a total order $\prec$ on the set of monomials
(disregarding coefficients) of $R$ satisfying
\begin{itemize}
\item \emph{multiplicativity} --- if $u,v,w$ are any monomials in
    $R$, then $u \prec v$ implies $wu \prec wv$;
\item \emph{well-ordering} --- if $\mathcal{F}$ is a collection of
monomials, then $\mathcal{F}$ has a smallest element under $\prec$.
\end{itemize}
Once a term ordering has been defined, one can then define the
notion of the \emph{leading term} of a polynomial, similar to the
univariate case.  For example, one defines the \emph{pure lexicographic order}
on $k[X,Y,Z]$ given by $Z \prec Y \prec X$ to be the ordering where
$X^a Y^b Z^c \prec X^d Y^e Z^f$ if and only if $(a,b,c) < (d,e,f)$
lexicographically.  With this ordering, an example of a polynomial
with its monomials in decreasing order is
\[ X^3 + X^2 Y + X^2 Z^7 + Y^9 + 1 \,; \]
its \emph{leading term} is $X^3=X^3Y^0Z^0$, and its \emph{trailing terms}
are $X^2 Y$, $X^2 Z^7$, $Y^9$ and $1$.

Given an ideal $I$, a \emph{Gr\"obner basis} $\mathcal{B}$ for $I$
is a set of
polynomials $g_1, g_2, \ldots, g_k$ such that the ideal generated
by the leading terms of the $g_i$ is precisely the initial ideal
of $I$, defined to be the set of leading
terms of polynomials in $I$.  It can be shown that
$\mathcal{B}$ generates $I$.  Furthermore, we say that
$\mathcal{B}$ is a \emph{reduced Gr\"obner basis} if
\begin{itemize}
\item the coefficient of each leading term in $\mathcal{B}$ is 1;
\item the leading terms of $\mathcal{B}$ are a \emph{minimal} set of
generators for the initial ideal of $B$; and
\item no trailing terms of $\mathcal{B}$ appear in the initial ideal
of $I$.
\end{itemize}
Once a term order has been chosen, reduced Gr\"obner bases are unique.
Note that in general, there are many term orderings for a polynomial ring $R$;
the computational difficulty of a computation involving Gr\"obner bases
is often highly sensitive to the choice of term ordering used.

Having established these preliminaries, we turn our attention
to solving a system of equations given by the commutative image of a
context free grammar.
Suppose we have a context-free grammar in the non-terminals
$S, N_1, N_2, \ldots, N_n$.  For each non-terminal $N$, let
$f_N$ also denote the generating function enumerating the language
generated by $N$.  Taking $k$ to be the field of formal
Laurent series $\mathbb{Q}((x))$, the Chomsky-Sch\"utzenberger theorem
implies $f_N \in k$ for every non-terminal $N$.  
Furthermore, by taking
the commutative image of the context-free grammar, we obtain
a sequence of polynomials $p_S, p_{N_1}, \ldots, p_{N_n}$, where
for every non-terminal $N$, the polynomial relation $p_N$ is the
commutative image of the derivation rule for $N$.  Note that
every such polynomial is in the polynomial ring
$\left(\mathbb{Z}[x]\right)[S, N_1, N_2, \ldots, N_n]$.

It follows from the definitions that for every non-terminal $N$,
\[ p_N(f_S, f_{N_1}, f_{N_2}, \ldots, f_{N_n}) = 0 \,; \]
that is, the $(n+1)$-tuple $(f_S, f_{N_1}, f_{N_2}, \ldots, f_{N_n})$
is a zero of the polynomial $p_N$.  Since this holds for every
non-terminal $N$, we can equivalently say that
$(f_S, f_{N_1}, f_{N_2}, \ldots, f_{N_n})$ is in the vanishing set
$V(I)$, where $I$ is generated by the polynomials
$p_S, p_{N_1}, p_{N_2}, \ldots, p_{N_n}$.

Our aim is to determine an algebraic equation satisfied by the power
series $f_S$.  To do this, we find a Gr\"obner basis $\mathcal{B}$ for $I$,
using an \emph{elimination ordering} on the indeterminate $S$.
The defining property of any such term ordering is that the monomials
involving only the indeterminate $S$ are strictly smaller than the
other monomials; namely, those involving at least one of
$N_1, N_2, \ldots, N_n$.  By the Chomsky-Sch\"utzenberger theorem
and the properties of Gr\"obner bases, the smallest polynomial $p$
in $\mathcal{B}$ will be a univariate polynomial in the indeterminate
$S$.  Since $p \in I$, and $(f_S, f_{N_1}, f_{N_2}, \ldots, f_{N_n})$
is in the vanishing set $V(I)$, we see that $p(f_S) = 0$; that is,
$p = 0$ is an algebraic equation satisfied by $f_S$.  (Note that in
previous sections, we simply use $S$ to denote $f_S$.)

As an example, we use Maple 13 to compute such an algebraic equation
for the example grammar in the previous section.  We give the commands,
followed by the produced output.
The commutative image of the grammar is entered as a list of
polynomials, given by \\
\verb!> eqs := [ -S + M + U, -M + x^2*M^2 + 1, -U + S*x + x^2*M*U ];! \\
\[ \mathit{eqs} := [-S+M+U,-M+{x}^{2}{M}^{2}+1,-U+Sx+{x}^{2}MU] \,. \]
Maple provides an elimination ordering called \verb!lexdeg!; to
compute a reduced Gr\"obner basis using this ordering, we enter the
command \\
\verb!> Groebner[Basis](eqs, lexdeg([M, U], [S]));! \\
\[ [1+ \left( -1+2\,x \right) S+ \left( -x+2\,{x}^{2} \right) {S}^{2},1+
 \left( -1+x \right) S+Ux,-1+ \left( 1-2\,x \right) S+Mx] \,. \]
The algebraic equation satisfied by $S$ is the first polynomial in
this set: \\
\verb!> algeq := %[1];! \\
\[ \mathit{algeq} := 1+ \left( -1+2\,x \right) S+ \left( -x+2\,{x}^{2} \right) {S}^{2} \,. \]
To compute the Laurent series zeros of $S$ using this polynomial,
we solve for $S$ and expand the solutions as Laurent series in
the indeterminate $x$: \\
\verb!> map(series, [solve(algeq, S)], x);!
\[ [(-{x}^{-1}-1-x-2\,{x}^{2}-3\,{x}^{3}-6\,{x}^{4}-10\,{x}^{5}+O \left( {
x}^{6} \right) ),(1+x+2\,{x}^{2}+3\,{x}^{3}+6\,{x}^{4}+10\,{x}^{5}+O
 \left( {x}^{6} \right) )] \,. \]
Our desired power series solution is the second entry in the above
returned list.
%\end{document}

%% file: analysis.tex
%\documentclass{article}
%\usepackage{amsmath, amssymb, amsthm}
%\newtheorem{theorem}{Theorem}

%\title{Asymptotic analysis}

%\begin{document}

If $L$ is a context-free language having an unambiguous grammar
and $f(x) = \sum a_n x^n$ is the
formal power series enumerating it, then $f(x)$ is algebraic over
$\mathbb{Q}(x)$ by Theorem~\ref{Chomsky-Schutzenberger}.
The previous section gave a procedure for computing an algebraic equation
satisfied by $f$; that is, we are able to determine a non-trivial
polynomial $P(x,S) \in \mathbb{Z}[x,S]$ such that
$P(x, f(x)) = 0$.
This section describes how \emph{singularity analysis}
can be used to determine the asymptotic growth rate of the coefficients
$a_n$.  We sketch some of the requisite notions from complex analysis
and provide a glimpse of the underlying theory;
more details can be found in Flajolet and
Sedgewick~\cite{Flajolet&Sedgewick:2009}.

The usefulness in considering complex analysis is that the formal
power series $f(x)$, defined purely combinatorially, can be viewed
as a function defined on an appropriate open subset of the complex plane
$\mathbb{C}$.  Such a function is called \emph{holomorphic} or
\emph{(complex) analytic}; this reinterpretation of $f(x)$ allows
us to apply theorems from complex analysis in order to derive bounds on
the asymptotic growth rate of the $a_n$ far tighter than what we could do with
purely combinatorial reasoning.

Indeed, assume that $L$ is an infinite context-free language ---
then there exists a real number $0 < R \leq 1$ called
the \emph{radius of convergence} for $f(x)$.  The defining properties
of $R$ are that:
\begin{itemize}
\item if $z$ is a complex number with $|z| < R$, then the infinite sum
$a_0 + a_1 z + a_2 z^2 + a_3 z^3 + \cdots$ converges; and
\item if $z$ is a complex number with $|z| > R$, then the infinite sum
$a_0 + a_1 z + a_2 z^2 + a_3 z^3 + \cdots$ diverges.
\end{itemize}
We note that the definition says nothing about the convergence of
$\sum a_i z^i$ when $|z| = R$.  Thus, defining $U$ to be the open ball
of complex numbers $z$ satisfying $|z| < R$, we can reinterpret $f$
as an \emph{analytic function on $U$}.  The connection between the
asymptotic growth of the coefficients $a_n$ and the number $R$ is given
by two theorems.

\begin{theorem}[Hadamard]
Given any power series, $R$ is given by the explicit formula:
\[ R = \frac{1}{\limsup_{n \to \infty} |a_n|^{1/n}} \,. \]
\end{theorem}

The defining properties of $\limsup$ state that
\begin{itemize}
\item for any $\varepsilon > 0$, the relation
    $|a_n|^{1/n} < \frac{1}{R} + \varepsilon$
    holds for sufficiently large $n$; and
\item for any $\varepsilon > 0$, the relation
    $|a_n|^{1/n} > \frac{1}{R} - \varepsilon$
    holds for infinitely many $n$.
\end{itemize}
For our situation in particular, this implies that up to a sub-exponential
factor, $a_n$ grows asymptotically like $1/R^n$.  (This implies that for any
$\varepsilon > 0$,
we have $a_n \in O((\frac{1}{R} + \varepsilon)^n)$ and
$a_n \notin O((\frac{1}{R} - \varepsilon)^n)$.

We note that Hadamard's formula applies to \emph{any} power series,
not just to generating functions of context-free languages.

An elementary
argument shows that our assumption that $L$ is infinite implies $R \leq 1$;
similarly, our assumption that $L$ is context-free (and thus algebraic)
implies $R > 0$.
(The argument for showing $R > 0$ is harder, and is sketched here for
those familiar with complex analysis.  The algebraic curve
given by $P(z,y) = 0$ determines $d$ branches around $z = 0$ and the
power series $f(x) = \sum_n a_n x^n$ must be associated with one such branch.
Since the exponents of $f(x)$ are non-negative integers, this
must be an analytic branch at $0$; hence, $f(x)$ determines an analytic
function at $0$ and must have positive radius of convergence.)

The second theorem describes the convergence of the power series
$f(x)$ on the circle given by $|z| = R$.  A \emph{dominant singularity} for
$f(x)$ is a point $z_0$ on this circle such that the sum
$\sum a_n z_0^n$ diverges; the following result says that a
positive (real-valued) dominant singularity always exists.

\begin{theorem}[Pringsheim]
Let $f(x) = \sum_n a_n x^n$ be a power series with radius of convergence
$R > 0$.  If the coefficients $a_n$ are all non-negative, then
$R$ is a dominant singularity for $f(x)$.
\end{theorem}

The benefit of Pringsheim's theorem is that, for the sake of determining
$R$, it suffices to examine the positive real line for the singularities
of $f(x)$ considered as a \emph{function}, not just as a power series.
We make this more precise now, by introducing the concept of a
\emph{multi-valued function}.

Suppose that the power series $f(x)$ is algebraic of degree $d$ over
$\mathbb{Q}(x)$ --- under the assumption that $P$ is irreducible,
this means that the degree of the polynomial
$P(x,S) \in \mathbb{Z}[x,S]$ in the variable $S$ is $d$, and we may
write
\[ P = q_n S^n + q_{n-1} S^{n-1} + q_{n-2} S^{n-2} + \cdots + q_0 \,, \]
where each $q_i$ is a polynomial in $\mathbb{Z}[x]$ and $q_n$ is non-zero.
(If $P$ is reducible, factor it and replace it by an appropriate
irreducible factor.)

If we work in the algebraically closed \emph{Puiseux series field}
$\bigcup_{n \geq 1} \mathbb{C}((x^{1/n}))$, we obtain $d$ roots of
$P(x,S) = 0$, say, $g_1(x), g_2(x), \ldots, g_d(x)$, one of which coincides
with $f(x)$.  In general, these roots will not be power series with
non-negative integer coefficients, but instead will be more generalized
power series with complex coefficients and (possibly negative)
fractional exponents.

Let $D(x) \in \mathbb{Z}[x]$ be the \emph{discriminant} of $P$ with respect
to the variable $S$; this is readily computed via the formula
\[ D = \frac{{(-1)}^{n(n-1)/2}}{q_n} \cdot
    \operatorname{Res}(P, \frac{\partial}{\partial S} P, S) \,. \]
Here, $\operatorname{Res}$ denotes the \emph{resultant} of two polynomials,
defined to be the determinant of a matrix whose entries are given by
the coefficients of the polynomials.
The theoretical importance of $D$ is that it satisfies the identity
\[ D(x) = q_n^{2(n-1)} \prod_{i \neq j} \left(g_i(x) - g_j(x) \right) \,. \]
Define the \emph{exceptional set} $\Xi$ of $P$ to be the complex zeros
of $D$; note that this is a finite set.  For every point $z$ in the complement
$\mathbb{C} \setminus \Xi$, where $D$ does not vanish, there exist $d$
distinct solutions $y$ to the equation $P(z,y) = 0$.  Furthermore, the
$d$ distinct solutions vary continuously with $z$, and a locally continuous
choice of solutions locally determines a \emph{branch} (which is locally
an analytic function) of the
algebraic curve cut out by $P(z,y) = 0$; this is how a
\emph{multi-valued function} arises.

On the open set $U$, which we have defined to be the set of points
$z$ satisfying $|z| < R$, one such branch is given by our initial
power series $f(x)$.  By Pringsheim's theorem, $f(x)$ diverges
at $R$; this shows that $f(x)$, considered as a analytic function on $U$,
has no analytic continuation to a function on an open set containing
$U \cup \{ R \}$.
According to the discussion above, this shows that $R$ must be
in the exceptional set $\Xi$.

We have given a method to calculate an upper bound for the
growth rate of the $a_n$; in particular, we have shown parts $(1)$ and $(2)$
of:
\begin{theorem}
Let $f(x) = \sum_n a_n x^n$ be a formal power series where
$a_n \geq 0$ for each $n$. Suppose $P(x,S) = 0$ is a
non-trivial algebraic equation satisfied by $f(x)$, and let $D$ be
the discriminant of $P$ with respect to $S$.  Then, exactly one of the positive
real roots $R$ of $D$ satisfies the following properties:
\begin{enumerate}
\item for any $\varepsilon > 0$, $a_n \in O((\frac{1}{R} + \varepsilon)^n)$;
\item for any $\varepsilon > 0$, $a_n \notin O((\frac{1}{R} - \varepsilon)^n)$;
and
\item if $D$ has no zero $z_0 \neq R$ such that $|z_0| = R$, then for
any $\varepsilon > 0$, $a_n \in \Omega((\frac{1}{R} - \varepsilon)^n)$.
\end{enumerate}
\end{theorem}
We remark that part $(3)$ is much more difficult to show; it is
implied by the stronger result that if $D$ has no zero $z_0 \neq R$ such that $|z_0| = R$, then there exists a polynomial $p$ such that $a_n \sim p(n)\cdot\left(\frac{1}{R}\right)^n$.

Given the list $\rho_1 < \rho_2 < \cdots < \rho_k$ of positive
real-valued elements of $\Xi$, there remains
the task of selecting which $\rho_j$ to use to provide an upper or
lower bound.
The bigger $j$ is, the better our upper bound will be; however, for this bound
to be valid, we must ensure that $\rho_j \leq R$.
For our purposes, we simply employ a boot-strapping method --- if is known
beforehand that $a_n \in O(n^s)$ for some $s$, then we simply
choose the minimal $j$ such that $1/\rho_j \leq s$; equivalently,
$\rho_j \geq 1/s$.  If this is not possible, we simply pick $j = 1$.
Similarly, for a lower bound, we choose the maximal $j$ such that
$\rho_j \leq 1/t$ if it is known that $a_n \in \Omega(n^t)$.
(With much more work, one can precisely identify $R$ ---
Flajolet and Sedgewick~\cite{Flajolet&Sedgewick:2009} describe an
algorithm ``Algebraic Coefficient Asymptotics'' that does this.)

As an illustration, we continue the Maple example in the previous
section to derive an asymptotic upper bound for the example grammar.
We first recall the algebraic equation satisfied by $S$: \\
\verb!> algeq;! \\
\[ 1+ \left( -1+2\,x \right) S+ \left( -x+2\,{x}^{2} \right) {S}^{2} \,. \]
We compute the discriminant $D$: \\
\verb!> d := discrim(algeq,S);! \\
\[ d := - \left( 2\,x+1 \right)  \left( -1+2\,x \right) \,. \]
The real roots of $D$ are given by: \\
\verb!> realroots := [fsolve(%)];!
\[ \mathit{realroots} := [ -0.5000000000, 0.5000000000 ] \,. \]
Finally, an upper bound is given by taking the inverse of the smallest positive real root: \\
\verb!> 1/min(op(select(type, realroots, positive)));! \\
\[ 2.000000000 \,. \]
Hence, $a_n \in O((2 + \varepsilon)^n)$ for any $\varepsilon > 0$.

%% file: lower.tex
We now turn to lower bounds on $R_k (n)$.
     In the unary case ($k = 1$), we can argue as follows:  consider
any subset of $\lbrace \epsilon, a, a^2, \ldots, a^{t-1} \rbrace$.  
Such a subset can be denoted by a regular expression of (ordinary) length
at most $t(t+1)/2$.  Since there are $2^t$ distinct subsets, this
gives a lower bound of $R_1 (n) \geq 2^{\sqrt{2n}-1}$.  
Similarly, when $k \geq 2$,
there are $k^n$ distinct strings of 
length $n$, so $R_k (n) \geq k^n$. 
These naive bounds can be improved somewhat using
a grammar-based approach.   

      Consider a regular expression of the form
$$w_1 (\epsilon + w_2 (\epsilon + w_3 (\epsilon + ... ))) $$
where the $w_i$ denote nonempty words.  Every distinct choice
of the $w_i$ specifies a distinct language.  Such expressions can be
generated by the grammar
\begin{eqnarray*}
S & \rightarrow&  Y \ | \ Y(\epsilon + S) \\
Y & \rightarrow&  aY \ | \ a \,, \quad a \in \Sigma 
\end{eqnarray*}
which has the commutative image
\begin{eqnarray*}
S & = & Y + YS x^4 \\
Y & = & kxY + kx \,.
\end{eqnarray*}
The solution to this system is
$$ S = \frac{kx}{1-kx-kx^5} \,.$$
Once again, the asymptotic behavior of the coefficients of the power series for $S$
depend on the zeros of $1-kx-kx^5$.  The smallest (indeed, the only)
real root is, asymptotically as
$k \rightarrow \infty$, given by
$$ \sum_{i\geq 0} \frac{ (-1)^i {\binom{5i}{i}} }{4i+1} k^{-(4i+1)} =
	\frac{1}{k} - \frac{1}{k^5} + \frac{5}{k^9} - \frac{35}{k^{13}} + \cdots .$$
The reciprocal of this series is
$$ \sum_{i \geq 0}  \frac{4 {\binom{5i+5}{i+1}} }{5 (5i+4)} k^{1-4i} =
k + \frac{1}{k^3} - \frac{4}{k^7} + \frac{26}{k^{11}} - \frac{204}{k^{15}} + \frac{1771}{k^{19}} - \cdots.$$
For $k = 1$ the only real root of $1-kx-kx^5$ is approximately $.754877666$
and for $k = 2$ it is about $.4756527435$.   Thus we have

\begin{theorem}
$R_1 (n) = \Omega(1.3247^n)$ and
$R_2 (n) = \Omega(2.102374^n)$.
\end{theorem}

\subsection{Trie representations for finite languages}\label{sec:lower-trie}
We will now improve these lower bounds. To this end, 
we begin with the simpler problem of counting the number
of finite languages that may be specified by regular
expressions without Kleene star of size $n$. Non-empty
finite languages not containing $\epsilon$ admit a standard 
representation via a trie structure; an example is given Fig.~\ref{trie}.
\begin{figure}
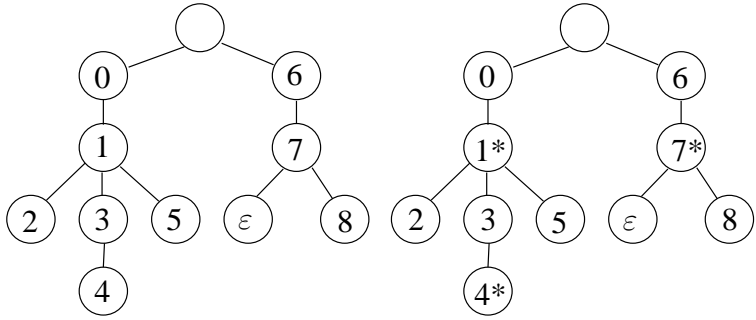

\centering
\subfigure[Representing the finite language \texttt{01(2+34+5)+67($\epsilon$+8)} as a trie.]{
\label{trie}
\input trie.pstex_t
}
\subfigure[Representing the infinite language 
\texttt{01*(2+34*+5)+67*($\epsilon$+8)} as a starred trie.]{
\label{startrie}
\input startrie.pstex_t
}
\caption{Example of a trie representation for a finite language (see 
Section~\ref{sec:lower-trie}) and 
of a starred trie representation 
for an infinite language (see Section~\ref{sec:lower-starred-trie}).}
\end{figure}

%Algorithm~\ref{ctree} below takes as input a
%finite non-empty language $L$ not containing
%$\epsilon$ and returns a trie in our desired
%format.
The words in such a language $L$ correspond to
the leaf nodes of the trie for $L$; moreover,
the concatenation of labels from the root
to a leaf node gives an expression for the
word associated with that leaf node.
For regular languages~$L$ and~$M$, we write
$M^{-1} L$ to denote the left quotient of $L$
by $M$; formally
$$M^{-1} L = \lbrace v \ : \ {\rm there\ exists} \  u \in M 
\ {\rm such\ that} \ uv \in L \rbrace .$$
If $M$ consists of a single word $w$, we also write
$w^{-1}L$ instead of $\{w\}^{-1}L$, and $w^{-n}L$ instead 
of $(w^{n})^{-1}L$.

%\burp{I have removed the pseudocode and the description of 
%how to construct a trie, since it is easily explained in 
%a few sentences of prose. I feel that the chapter has 
%too many floating bodies overall, and personally I do not 
%love pseudocode that much.
%}
%
%\begin{algorithm}
%\caption{CREATE-TRIE($L$)}
%\label{ctree}
%\begin{algorithmic}[1]
%\REQUIRE $\epsilon \not\in L$, $L \neq \emptyset$
%\STATE create a tree $T$ with an unlabelled root
%\FORALL{$a \in \Sigma$ such that $a^{-1} L \neq \emptyset$}
%\STATE add the subtree returned by CREATE-TRIE-HELP($\{a\}^{-1} L$, $a$)
%as a child of the root of $T$
%\ENDFOR
%\STATE return $T$
%\end{algorithmic}
%\end{algorithm}
%
%\begin{algorithm}
%\caption{CREATE-TRIE-HELP($L$, $a$)}
%\label{ctreehelp}
%\begin{algorithmic}[1]
%\STATE create a tree $T$ with a root labelled $a$
%\IF[need to create children]{$L \neq \{ \epsilon \}$}
%\FORALL{$b \in \Sigma$ such that $b^{-1} L \neq \emptyset$}
%\STATE add the subtree returned by CREATE-TRIE-HELP($\{b\}^{-1} L$, $b$)
%as a child of the root of $T$
%\ENDFOR
%\IF{$\epsilon \in L$}
%\STATE add a node labelled $\epsilon$ as a child of the
%root of $T$
%\ENDIF
%\ENDIF
%\STATE return $T$
%\end{algorithmic}
%\end{algorithm}

For notational convenience, we take our alphabet to
be $\Sigma = \{ a_0, a_1, \ldots, a_{k-1} \}$, where
$k \geq 1$ denotes our alphabet size.
A trie encodes the simple fact that  
each nonempty finite language~$L$ not 
containing~$\epsilon$ can be uniquely 
decomposed as $L = \bigcup_i a_iL_i$, 
where $L_i = a_i^{-1}L$, and the index~$i$ runs 
over all symbols~$a_i\in \Sigma$ such that~$L_i$ 
is nonempty. This factoring out of common prefixes resembles
Horner's rule (see e.g.~\cite[p.~486]{Knuth:1997:TAOCP2}) for evaluating 
polynomials. We develop lower bounds by specifying a 
context-free grammar that generates regular expressions 
with common prefixes factored out. 
In fact, the grammar is designed so that if~$r$ is
a regular expression generated by the grammar, then
the structure of~$r$ mimics that of the trie for
$L(r)$ --- nodes with a single child correspond to
concatenations, while nodes with multiple children
correspond to concatenations with a union, see 
Table~\ref{tab:grammar-trie}.

%\begin{eqnarray*}
%S & \to & Y \;|\; Z \\
%E & \to & Y \;|\; (Z) \;|\; (\epsilon+S) \\
%Y & \to & P_i \text{ for } 0 \leq i < k \\
%Z & \to & P_{n_0} + P_{n_1} + \cdots + {P_{n_t}}
%	\text{ where } 0 \leq n_0 < n_1 < \cdots < n_t < k
%	\text{ for } t > 0 \\
%P_i & \to & a_i \;|\; a_i E \text{ for } 0 \leq i < k
%\end{eqnarray*}
\begin{table}[ht]
{
\renewcommand{\arraystretch}{1.1}
\renewcommand{\tabcolsep}{2mm}
\centerline{   
\begin{tabular}{|rl|}
\hline
\hline
$S\to$ & $Y \mid Z$ \\
$E \to$ & $Y \mid (Z) \mid (\epsilon+S)$ \\
$Y \to$ & $P_i \text{ for } 0 \leq i < k $\\
$Z \to$ & $P_{n_0} + P_{n_1} + \cdots + {P_{n_t}}$\\
	& $\text{ where } 0 \leq n_0 < n_1 < \cdots < n_t < k
	\text{ for } t > 0 $\\
$P_i  \to$ & $a_i \mid a_i E \text{ for } 0 \leq i < k$\\
\hline
\hline
\end{tabular}
}}
\caption{A grammar for mimicking tries with regular expressions.}
\label{tab:grammar-trie}
\end{table}

The set of regular languages represented corresponds to
all non-empty finite languages over $\Sigma$ not containing
the empty string $\epsilon$.
We briefly describe the non-terminals:
\begin{description}
\item[$S$] generates all non-empty finite languages not
containing $\epsilon$. %--- this corresponds to 
%Algorithm~\ref{ctree}.

\item[$E$] generates all non-empty finite languages
containing at least one word other than $\epsilon$. %---
%this corresponds to line $2$ of Algorithm~\ref{ctreehelp}.

\item[$Y$] generates all non-empty finite languages
(not containing $\epsilon$) whose words all begin with
the same letter. %--- this corresponds to line $2$ of
%Algorithm~\ref{ctree} and line $3$ of
%Algorithm~\ref{ctreehelp} when the body of the
The \texttt{for} loop is executed only once.

\item[$Z$] generates all non-empty finite languages
(not containing $\epsilon$) whose words do not all begin
with the same letter. % --- this corresponds to line $2$ of
%Algorithm~\ref{ctree} and line $3$ of
%Algorithm~\ref{ctreehelp} when the body of the
%\texttt{for} loop is executed more than once.

\item[$P_i$] generates all non-empty finite languages
(not containing $\epsilon$) whose words all begin 
with~$a_i$. %--- this corresponds to line $1$ of
%Algorithm~\ref{ctreehelp}.
\end{description}

We remark that this grammar is unambiguous and that
no regular language is represented more than once;
this should be clear from the relationship between
regular expressions generated by the grammar and
their respective tries.

(Note that it is possible to slightly optimize this
grammar in the case of ordinary length to generate
expressions such as $\tt{0 + 00}$ in lieu of $\tt{0(\epsilon+0)}$,
but as it results in marginal improvements to the lower
bound at the cost of greatly complicating the grammar,
we do not do so here.)

Table~\ref{sf_lower} lists the lower bounds obtained
through this grammar.  In this table (and only this table),
each $\Omega(k^n)$ in the column corresponding to reverse polish notation
should be interpreted as ``not $O(k^n)$'' --- observe, for instance,
that all strings produced by our grammar for a unary alphabet have
odd reverse polish length.

\begin{table}
\center
\setlength{\tabcolsep}{0.1in}
\begin{tabular}{|r|lll|}
\hline
\hline
& ordinary & reverse polish & alphabetic \\ \hline
$1$ & $\Omega(1.3247^n)$ & $\Omega(1.2720^n)$ & $\Omega(2^n)$ \\
$2$ & $\Omega(2.5676^n)$ & $\Omega(2.1532^n)$ & $\Omega(6.8284^n)$ \\
$3$ & $\Omega(3.6130^n)$ & $\Omega(2.7176^n)$ & $\Omega(11.1961^n)$ \\
$4$ & $\Omega(4.6260^n)$ & $\Omega(3.1806^n)$ & $\Omega(15.5307^n)$ \\
$5$ & $\Omega(5.6264^n)$ & $\Omega(3.5834^n)$ & $\Omega(19.8548^n)$ \\
$6$ & $\Omega(6.6215^n)$ & $\Omega(3.9451^n)$ & $\Omega(24.1740^n)$\\
\hline
\hline
\end{tabular}
\caption{\label{sf_lower}Lower bounds for $R_k(n)$
with respect to size measure and alphabet size.}
\end{table}

\begin{remark}
Using the singularity analysis method explained in Section~\ref{sec:analysis}, 
these lower bounds were
obtained by boot-strapping off the trivial bounds of
$\Omega(k^n)$, $\Omega(k^{n/2})$ and $\Omega(k^n)$ for
the ordinary, reverse polish length and alphabetic width cases,
respectively.
\end{remark}

Before we generalize our approach to cover also 
infinite languages, we derive a formula showing how our 
lower bound on alphabetic width will increase along 
with the alphabet size~$k$.

%\burp{move the following into asymptotic analysis section?}

To this end, we first state a version of the Lagrange implicit function
theorem as a simplification of \cite[Theorem 1.2.4]{Goulden&Jackson:1983}.
If $f(x)$ is a power series in $x$, we write
$[x^n] f(x)$ to denote the coefficient of $x^n$ in $f(x)$;
recall that the {\em characteristic} of a ring~$R$ with 
additive identity $0$ and multiplicative identity $1$ is 
defined to be the smallest integer~$k$ such 
that $\sum_{i=1}^k 1 = 0$, or zero if there is no such~$k$. 

\begin{lemma}
Let $R$ be a commutative ring of characteristic zero
and take $\phi(\lambda) \in R[[\lambda]]$
such that $[\lambda^0] \phi$ is invertible.
Then there exists a unique formal power series
$w(x) \in R[[x]]$ such that $[x^0] w = 0$ and
$w = x \phi(w)$. For $n \geq 1$,
\[ [x^n] w(x) = \frac{1}{n} [\lambda^{n-1}] \phi^n(\lambda) \,. \]
\end{lemma}

Due to the simplicity of alphabetic width, the problem
of enumerating regular languages in this case may be
interpreted as doing so for rooted $k$-ary trees, where
each internal node is marked with one of two possible
colours. We thus investigate how our lower bound
varies with~$k$.

More specifically, consider a regular expression $r$
generated by the grammar from the previous section
and its associated trie.
Colour each node with a child labelled~$\epsilon$ black
and all other nodes white. After deleting all nodes
marked~$\epsilon$, call the resultant tree~$T(r)$.
This operation is reversible, and shows that we may
put the expressions of alphabetic width~$n$ in
correspondence with the $k$-ary rooted trees with
$n+1$ vertices where every non-root internal node
may assume one of two colours. In order to estimate
the latter, we first prove a basic result. The first
half of the following lemma is also 
found in~\cite[p.~68]{Flajolet&Sedgewick:2009}.

\begin{lemma}
There are $\frac{1}{n} \binom{kn}{n-1}$ $k$-ary trees
of $n$ nodes.
Moreover, the expected number of leaf nodes among
$k$-ary trees of~$n$ nodes is asymptotic to
$(1 - 1/k)^k n$ as $n \to \infty$.
\end{lemma}

\begin{proof}
Fix $k \geq 1$.
For $n \geq 1$, let $a_n$ denote the number of
$k$-ary rooted trees with $n$ vertices and consider
the generating series:
\[ f(x) = \sum_{n \geq 1} a_n x^n \,. \]
By the recursive structure of $k$-ary trees, we
have the recurrence:
\[ f(x) = t(1 + f(x))^k \,. \]
Thus, by the Lagrange implicit function 
theorem, we have
%\begin{eqnarray*}
%a_n & = & [x^n] f(x) \\
%& = & \frac{1}{n} [\lambda^{n-1}] (1 + \lambda)^{kn} \\
%& = & \frac{1}{n} \binom{kn}{n-1} \,.
%\end{eqnarray*}
\[
a_n  =  [x^n] f(x) 
 =  \frac{1}{n} [\lambda^{n-1}] (1 + \lambda)^{kn} 
 =  \frac{1}{n} \binom{kn}{n-1} \,.
\]
We now calculate the number of leaf nodes among
all $k$-ary rooted trees with $n$ vertices.
Let $b_{n,m}$ denote the number of $k$-ary
rooted trees with $n$ vertices and $m$ leaf nodes
and $c_n$ the number of leaf nodes among all
$k$-ary rooted trees with $n$ vertices.
Consider the bivariate generating series:
\[ g(x,y) = \sum_{n,m \geq 1} b_{n,m} x^m y^n \,. \]
By the recursive structure of $k$-ary trees, we
have the recurrence:
\[ g(x,y) = y(x - 1 + (1 + g(x,y))^k) \,. \]
The Lagrange implicit function theorem once again
yields
\begin{eqnarray*}
c_n & = & \left. \frac{\partial}{\partial x} [y^n] g(x,y)
	\right|_{x = 1} \\
& = & \left. \frac{\partial}{\partial x} \frac{1}{n}
	[\lambda^{n-1}] (x - 1 + (1 + g(x,y))^k)^n
	\right|_{x = 1} \\
& = & \frac{1}{n} [\lambda^{n-1}] \left.
	\frac{\partial}{\partial x} (x - 1 + (1 + \lambda)^k)^n
	\right|_{x = 1} \\
& = & [\lambda^{n-1}] (1 + \lambda)^{k(n-1)} \\
& = & \binom{k(n - 1)}{n - 1} \,.
\end{eqnarray*}
Thus, the expected number of leaf nodes among $n$-node
trees is, as  $n \to \infty$ while having $k$ fixed,
\[\frac{c_n}{a_n}  =  \frac{n \binom{k(n - 1)}{n - 1}}
	{\binom{kn}{n-1}} 
 \sim n \left( \frac{k-1}{k} \right)^k.	
\]
%we probably do not need an explicit computation here.
%\begin{eqnarray*}
%\frac{c_n}{a_n} & = & \frac{n \binom{k(n - 1)}{n - 1}}
%	{\binom{kn}{n-1}} \\
%& = & \frac{n(kn-n+1)(kn-n)\cdots(kn-k-n+2)}
%	{(kn)(kn-1)\cdots(kn-k+1)} \\
%& \sim & n \left( \frac{k-1}{k} \right)^k
%	\text{ as } n \to \infty \text{ for fixed } k.
%\end{eqnarray*}
\end{proof}

We wish to find a bound on the expected number of
subsets of non-root internal nodes among all $k$-ary
rooted trees with $n$ nodes, where a subset
corresponds to those nodes marked black.
Fix $k \geq 2$.
Since the map $x \mapsto 2^x$ is convex, for
every $\varepsilon > 0$ and sufficiently large $n$,
Jensen's inequality (e.g., \cite[Thm.\ 3.3]{Rudin:1966})
applied to the lemma above implies the following
lower bound on the number of subsets:
\[ 2^{(1 - (1 - 1/k)^k - \varepsilon) n} \,. \]
Since $-(1 - 1/k)^k > -1/e$ for $k \geq 1$,
we may choose $\varepsilon > 0$ such that
\[ -(1 - 1/k)^k - \varepsilon > -1/e \,. \]
This yields a lower bound of $2^{(1-1/e)n}$. 

Assuming $k \geq 2$ fixed, we now estimate $\binom{kn}{n-1}$.
By Stirling's formula, we have, as $n \to \infty$,

\begin{eqnarray*}
%we probably do not need an explicit computation here.
\binom{kn}{n-1} %& = & 
%\frac{(kn)!}{n!\;((k-1)n)!} \frac{n}{(k-1)n + 1} \\
%& \sim & \frac{\sqrt{2 \pi kn} \; (kn/e)^{kn}}
%	{\sqrt{2 \pi n} \; (n/e)^n \;
%	\sqrt{2 \pi (k-1)n} \; ((k-1)n/e)^{(k-1)n}}
%	\frac{n}{(k-1)n + 1} \\
& = & \Theta\left(\left( \frac{k^k}{(k-1)^{k-1}} \right)^n \right) \,.
\end{eqnarray*}

Putting our two bounds together, we have the following
lower bound on the number of star-free
regular expressions of alphabetic width $n$,
when $n \to \infty$ while keeping $k$ fixed:
\[ \Omega\left( \left( \frac{ 2^{(1-1/e)} k^k}
	{(k-1)^{k-1}} \right)^n \right). \]

\subsection{Trie representations for some infinite regular languages}\label{sec:lower-starred-trie}

We now turn our attention to enumerating regular languages
in general; that is, we allow for regular expressions with
Kleene stars.

Our grammars for this section are based on the those for
the star-free cases.
Due to the difficulty of avoiding specifying duplicate
regular languages, we settle for a ``small'' subset of
regular languages.
For simplicity, we only consider taking the Kleene
star closure of singleton alphabet symbols, and we impose some 
further restrictions.

Recall the trie representation of a star-free
regular expression written in our common prefix
notation.
With this representation, we may mark nodes with stars
while satisfying the following conditions:
\begin{itemize}
\item each starred symbol must have a non-starred parent
	other than the root;
\item a starred symbol may not have a sibling or an
	identically-labelled parent (disregarding
	the lack of star) with its own sibling; and
\item a starred symbol may not have an identically-labelled
	child (disregarding the lack of star).
\end{itemize}

The first condition eliminates duplicates such as
\[ \texttt{0*11*0*1*0*} \leftrightarrow
\texttt{0*1*0*11*0*} \,; \]
the second eliminates those such as
\[ \texttt{01*} \leftrightarrow
\texttt{0($\epsilon$+11*)} \text{ and }
\texttt{0(1+2*1)} \leftrightarrow \texttt{02*1} \]
and the third eliminates those such as
\[ \texttt{0*0} \leftrightarrow \texttt{00*} \,. \]
In this manner, we end up with starred tries such as
in Fig.~\ref{startrie}.
Algorithm~\ref{cstree} illustrates how to recreate such
a starred trie from the language it specifies.
\begin{algorithm}
\caption{STAR-TRIE($L$)}
\label{cstree}
\begin{algorithmic}[1]
\REQUIRE $\epsilon \not\in L$, $L \neq \emptyset$
\STATE create a tree $T$ with unlabelled root
\FORALL{$a \in \Sigma$ such that $a^{-1} L \neq \emptyset$}
\STATE append STAR-TRIE-HELP($a^{-1} L$, $a$) below the root of $T$
\ENDFOR
\STATE return $T$
\end{algorithmic}
\end{algorithm}

\begin{algorithm}
\caption{STAR-TRIE-HELP($L$, $a$)}
\label{cstreehelp}
\begin{algorithmic}[1]
\STATE create a tree $T$ with root labelled $a$
\FORALL{$b \in \Sigma$ s.t. $b^{-1} L \neq \emptyset$}
  \IF[need a child labelled $b^*$]{
        $\left(b^{-n}L\right) \cap (\epsilon
	+ (\Sigma \setminus \{b\}) \Sigma^*) \neq \emptyset$
	for all $n \geq 0$}
    \STATE append a new $b^*$-node below the root of $T$
    % begin: starred child's children
    \IF[$b^*$ will be an internal node]{$L \neq b^*$}
      \FORALL[determine children of $b^*$]{$c \in \Sigma\setminus\{b\}$ such that
	  $c^{-1} L \neq \emptyset$}
        \STATE append STAR-TRIE-HELP($c^{-1} L$, $c$) below
	    the $b^*$-node
      \ENDFOR%determine children
      
      \IF{$b \in L$}
        \STATE append a new $\epsilon$-node below the $b^*$-node
      \ENDIF%$b \in L$
    \ENDIF%$b^*$ will be an internal node
  %\ENDIF%need a starred...
  \ELSE[need a child labelled $b$]
    \STATE append STAR-TRIE-HELP($b^{-1} L$, $b$) below
      the root of $T$
  \ENDIF%need an unstarred
\ENDFOR%$b \in \Sigma$
\IF{$\epsilon \in L$ and the root of $T$ has at least one unstarred child}
  \STATE append a new $\epsilon$-node below the root of $T$
\ENDIF
\STATE return $T$
\end{algorithmic}
\end{algorithm}

Let $T$ be any starred trie satisfying the conditions
above. Then $T$ represents a regular expression, which
in turn specifies a certain language.
We now show that when the algorithm is run with that
language as input, it returns the trie $T$ by arguing
that at each step of the algorithm when a particular node
(matched with language $L$ if the root and $aL$ otherwise) is being processed,
the children are correctly reconstructed.

We first consider children of the root.
By the original trie construction (for finite
languages without $\epsilon$), no such children
may be labelled $\epsilon$.
Thus, by the first star condition, the only
children may be unstarred alphabet symbols.
Thus, line~$2$ of Algorithm~\ref{cstree}
suffices to find all children of the root
correctly.

Now consider a non-root internal node, say
labelled $a$.
By the third star condition, a starred node
may not have a child labelled with the same
alphabet symbol, so if $a$ has a child
labelled $b^*$, then
\begin{align}\label{condition:bstar} 
\left(b^n\right)^{-1} L \cap (\epsilon
        + (\Sigma \setminus \{b\}) \Sigma^*)
	\mbox{ is non-empty for all $n \geq 0$}. 
\end{align}
Conversely, by the second condition, a
starred node may not have an
identically-labelled parent that 
has~$\epsilon$ 
as a sibling, so if~\eqref{condition:bstar}
%\[ \left(b^n\right)^{-1} L \cap (\epsilon
%        + (\Sigma \setminus \{b\}) \Sigma^*)
%	\mbox{ is non-empty for all $n \geq 0$}, \]
holds,
then $a$ must 
have a child labelled $b^*$.
By the second star condition, a starred
node may not have siblings, so the algorithm
need not check for other children once a
starred child is found.
This shows that line~$3$ of Algorithm~\ref{cstreehelp}
correctly identifies all starred children of~$a$.
Assuming~$a$ has a starred child~$b^*$, then
by the third condition, line~$6$ of
Algorithm~\ref{cstreehelp} correctly recovers
all children of~$b^*$.
All remaining children of~$a$ have no stars, and
line~$14$ of Algorithm~\ref{cstreehelp} suffices
to find all children labelled with $a\in \Sigma$; 
the special case of an 
$\epsilon$-child below~$a$ is covered by line~$17$.

\begin{table}
{
\renewcommand{\arraystretch}{1.1}
\renewcommand{\tabcolsep}{2mm}
\centerline{   
\begin{tabular}{|rl|
}
\hline
\hline
$S  \to$ & $Y \mid Z $\\
\hline\hline
$E  \to$ & $Y \mid (Z) \mid (\epsilon+Y') \mid (\epsilon+Z) $\\
$E_i  \to$ & $Y_i \mid (Z_i) \mid (\epsilon+Y_i') \mid (\epsilon+Z_i)
\text{ for } 0 \leq i < k $\\
\hline\hline
$Y  \to$ & $P_i \text{ for } 0 \leq i < k $\\
$Y'  \to$ & $P_i' \text{ for } 0 \leq i < k $\\
$Y_i  \to$ & $P_j \text{ for } 0 \leq i,j < k
	\text{ and } i \neq j $\\
$Y_i'  \to$ & $P_j' \text{ for } 0 \leq i,j < k
	\text{ and } i \neq j $\\
\hline\hline
$Z  \to$ & $P_{n_0}' + P_{n_1}' + \cdots + P_{n_t}'$\\
	&$\text{ where } 0 \leq n_0 < n_1 < \cdots < n_t < k
	\text{ for } t > 0 $\\
$Z_i  \to$ & $P_{n_0}' + P_{n_1}' + \cdots + P_{n_t}'$\\
	&$\text{ as above, but with }
	n_j \neq i \text{ for all } 0 \leq j \leq t$ \\
\hline\hline
$P_i  \to$ & $a_i \mid a_i E \mid a_ia_j^* \mid a_ia_j^* E_j$\\
	& $\text{ for } 0 \leq i,j < k$ \\
$P_i'  \to$ & $a_i \mid a_i E \mid a_ia_j^* \mid a_ia_j^* E_j$\\
	& $\text{ for } 0 \leq i,j < k \text{ and } i \neq j$\\
\hline
\hline
\end{tabular}
}%end of centerline's scope
}
\caption{A grammar generating all regular expressions 
meeting all three star conditions.}\label{tab:grammar-trie2}
\end{table}

%\begin{eqnarray*}
%S & \to & Y \;|\; Z \\
%\\
%E & \to & Y \;|\; (Z) \;|\; (\epsilon+Y') \;|\; (\epsilon+Z) \\
%E_i & \to & Y_i \;|\; (Z_i) \;|\; (\epsilon+Y_i') \;|\; (\epsilon+Z_i)
%\text{ for } 0 \leq i < k \\
%\\
%Y & \to & P_i \text{ for } 0 \leq i < k \\
%Y' & \to & P_i' \text{ for } 0 \leq i < k \\
%Y_i & \to & P_j \text{ for } 0 \leq i,j < k
%	\text{ and } i \neq j \\
%Y_i' & \to & P_j' \text{ for } 0 \leq i,j < k
%	\text{ and } i \neq j \\
%\\
%Z & \to & P_{n_0}' + P_{n_1}' + \cdots + P_{n_t}'
%	\text{ where } 0 \leq n_0 < n_1 < \cdots < n_t < k
%	\text{ for } t > 0 \\
%Z_i & \to & P_{n_0}' + P_{n_1}' + \cdots + P_{n_t}'
%	\text{ as above, but with }
%	n_j \neq i \text{ for all } 0 \leq j \leq t \\
%\\
%P_i & \to & a_i \;|\; a_i E \;|\; a_ia_j^* \;|\; a_ia_j^* E_j
%	\text{ for } 0 \leq i,j < k \\
%P_i' & \to & a_i \;|\; a_i E \;|\; a_ia_j^* \;|\; a_ia_j^* E_j
%	\text{ for } 0 \leq i,j < k \text{ and } i \neq j
%\end{eqnarray*}

We give a grammar that generates expressions meeting
these conditions in Table~\ref{tab:grammar-trie2}.
As before, we take our alphabet to be
$\Sigma = \{ a_0, a_1, \ldots, a_{k-1} \}$.
We describe the roles of the non-terminals of 
the grammar in Table~\ref{tab:grammar-trie2}.

\begin{description}
\item[$S$] generates all expressions --- this corresponds
	to Algorithm~\ref{cstree}.

\item[$E, E_i$] generate expressions that may be concatenated
	to non-starred and starred alphabet symbols, respectively.
	The non-terminal $E$ corresponds to lines $2$ and $13$ while
	$E_i$ corresponds to line $5$ of Algorithm~\ref{cstreehelp}.
	These act the same as $S$ except for the introduction of
	parentheses to take precedence into account and restriction
	that no prefixes of the form $\epsilon + aa^*$ are
	generated, used to implement the second condition.

	Additionally, $E_i$ has the restriction that its first
	alphabet symbol produced may not be $a_i$ --- this
	is used to implement the third condition.

\item[$Y, Y', Y_i, Y_i'$] generate expressions whose prefix
	is an alphabet symbol.
	As a whole, these non-terminals correspond to
	Algorithm~\ref{cstreehelp}, and may be considered
	degenerate cases of $Z$ and $Z_i$; that is, trivial
	unions.

	The tick-mark signifies that expressions of the form
	$aa^*$ for $a \in \Sigma$ are disallowed, used to
	implement the second condition.
	The subscripted $i$ signifies that the initial alphabet
	symbol may not be $a_i$, used to implement the third
	condition.
	
\item[$Z, Z_i$] generate non-trivial unions of expressions
	beginning with distinct alphabet symbols --- $Z$
	corresponds to line $2$ of Algorithm~\ref{cstree}
	and line $13$ of Algorithm~\ref{cstreehelp}, while
	$Z_i$ corresponds to line $5$ of
	Algorithm~\ref{cstreehelp}.

	The subscripted $i$ signifies that none of initial
	alphabet symbols may be $a_i$, used to implement
	the third condition.

\item[$P_i, P_i'$] generate expressions beginning with
	the specified alphabet symbol $a_i$. They correspond
	to line $1$ of Algorithm~\ref{cstreehelp}.

	The tick-mark signifies that expressions may not
	have the prefix $a_ia_i^*$, used to implement
	the second condition.
\end{description}

Since the algorithm correctly returns a trie when run on
the language represented by the trie, the correspondence
between the algorithm and the grammar gives us the
following result.

\begin{theorem}
The grammar above is unambiguous and the generated regular
expressions represent distinct regular languages.
\end{theorem}

Table~\ref{gen_lower} lists the improved lower bounds for
$R_k(n)$. These lower bounds were obtained 
via singularity analysis, as explained in Section~\ref{sec:analysis},
boot-strapping off the bounds in Table~\ref{sf_lower}.\footnote{The Maple worksheets used to 
derive these bounds can be accessed at the second author's 
personal homepage via 
\url{http://math.stanford.edu/~jlee/automata/}}

\begin{table}
\center
\setlength{\tabcolsep}{0.125in}
\begin{tabular}{|r|lll|}
\hline
\hline
$k$ & ordinary & reverse polish & alphabetic \\ \hline
$1$ & $\Omega(1.3247^n)$ & $\Omega(1.2720^n)$ & $\Omega(2^n)$ \\
$2$ & $\Omega(2.7799^n)$ & $\Omega(2.2140^n)$ & $\Omega(7.4140^n)$ \\
$3$ & $\Omega(3.9582^n)$ & $\Omega(2.8065^n)$ & $\Omega(12.5367^n)$ \\
$4$ & $\Omega(5.0629^n)$ & $\Omega(3.2860^n)$ & $\Omega(17.6695^n)$ \\
$5$ & $\Omega(6.1319^n)$ & $\Omega(3.6998^n)$ & $\Omega(22.8082^n)$ \\
$6$ & $\Omega(7.1804^n)$ & $\Omega(4.0693^n)$ & $\Omega(27.9500^n)$\\
\hline
\hline
\end{tabular}
\caption{\label{gen_lower}Improved lower bounds for $R_k(n)$
with respect to size measure and alphabet size.}
\end{table}

%% file: trie.pstex_t
\begin{picture}(0,0)%
\epsfig{file=trie.pstex}%
\end{picture}%
\setlength{\unitlength}{3947sp}%
\begingroup\makeatletter\ifx\SetFigFont\undefined%
\gdef\SetFigFont#1#2#3#4#5{%
  \reset@font\fontsize{#1}{#2pt}%
  \fontfamily{#3}\fontseries{#4}\fontshape{#5}%
  \selectfont}%
\fi\endgroup%
\begin{picture}(2266,1964)(593,-1568)
\put(1155,-129){\makebox(0,0)[lb]{\smash{\SetFigFont{12}{14.4}{\rmdefault}{\mddefault}{\updefault}0}}}
\put(1156,-564){\makebox(0,0)[lb]{\smash{\SetFigFont{12}{14.4}{\rmdefault}{\mddefault}{\updefault}1}}}
\put(706,-1028){\makebox(0,0)[lb]{\smash{\SetFigFont{12}{14.4}{\rmdefault}{\mddefault}{\updefault}2}}}
\put(1156,-1471){\makebox(0,0)[lb]{\smash{\SetFigFont{12}{14.4}{\rmdefault}{\mddefault}{\updefault}4}}}
\put(2348,-113){\makebox(0,0)[lb]{\smash{\SetFigFont{12}{14.4}{\rmdefault}{\mddefault}{\updefault}6}}}
\put(2356,-578){\makebox(0,0)[lb]{\smash{\SetFigFont{12}{14.4}{\rmdefault}{\mddefault}{\updefault}7}}}
\put(2663,-1028){\makebox(0,0)[lb]{\smash{\SetFigFont{12}{14.4}{\rmdefault}{\mddefault}{\updefault}8}}}
\put(1156,-1019){\makebox(0,0)[lb]{\smash{\SetFigFont{12}{14.4}{\rmdefault}{\mddefault}{\updefault}3}}}
\put(1604,-1021){\makebox(0,0)[lb]{\smash{\SetFigFont{12}{14.4}{\rmdefault}{\mddefault}{\updefault}5}}}
\put(2048,-1013){\makebox(0,0)[lb]{\smash{\SetFigFont{12}{14.4}{\rmdefault}{\mddefault}{\updefault}$\epsilon$}}}
\end{picture}

%% file: startrie.pstex_t
\begin{picture}(0,0)%
\epsfig{file=startrie.pstex}%
\end{picture}%
\setlength{\unitlength}{3947sp}%
\begingroup\makeatletter\ifx\SetFigFont\undefined%
\gdef\SetFigFont#1#2#3#4#5{%
  \reset@font\fontsize{#1}{#2pt}%
  \fontfamily{#3}\fontseries{#4}\fontshape{#5}%
  \selectfont}%
\fi\endgroup%
\begin{picture}(2266,1964)(593,-1568)
\put(1126,-586){\makebox(0,0)[lb]{\smash{\SetFigFont{12}{14.4}{\rmdefault}{\mddefault}{\updefault}1*}}}
\put(1126,-1486){\makebox(0,0)[lb]{\smash{\SetFigFont{12}{14.4}{\rmdefault}{\mddefault}{\updefault}4*}}}
\put(1156,-121){\makebox(0,0)[lb]{\smash{\SetFigFont{12}{14.4}{\rmdefault}{\mddefault}{\updefault}0}}}
\put(714,-1021){\makebox(0,0)[lb]{\smash{\SetFigFont{12}{14.4}{\rmdefault}{\mddefault}{\updefault}2}}}
\put(1163,-1021){\makebox(0,0)[lb]{\smash{\SetFigFont{12}{14.4}{\rmdefault}{\mddefault}{\updefault}3}}}
\put(1606,-1036){\makebox(0,0)[lb]{\smash{\SetFigFont{12}{14.4}{\rmdefault}{\mddefault}{\updefault}5}}}
\put(2349,-586){\makebox(0,0)[lb]{\smash{\SetFigFont{12}{14.4}{\rmdefault}{\mddefault}{\updefault}7*}}}
\put(2356,-128){\makebox(0,0)[lb]{\smash{\SetFigFont{12}{14.4}{\rmdefault}{\mddefault}{\updefault}6}}}
\put(2663,-1021){\makebox(0,0)[lb]{\smash{\SetFigFont{12}{14.4}{\rmdefault}{\mddefault}{\updefault}8}}}
\put(2048,-1021){\makebox(0,0)[lb]{\smash{\SetFigFont{12}{14.4}{\rmdefault}{\mddefault}{\updefault}$\epsilon$}}}
\end{picture}

%% file: upper.tex
Turning our attention back to upper bounds for $R_k(n)$, 
we develop grammars
for regular expressions such that every regular language
is represented by at least one shortest regular expression
generated by the grammar, where a regular expression~$r$ of
size~$n$ is said to be shortest if there is no 
expression~$r'$ of size less than $n$ with $L(r) = L(r')$. 

To this end, we consider certain ``normal forms'' 
for regular expressions, with the 
property that transforming a regular expression into normal form never
increases its size. Again, 
size may refer to one of the various measures introduced 
before. With such a normal form, it suffices to enumerate all 
regular expressions in normal form to obtain improved upper 
bounds on $R_k(n)$ for various measures.

\subsection{A grammar based on normalized regular expressions}

We begin with a simple approach, which will be further refined
later on. As concatenation and sum are associative, we consider 
them to be variadic operators taking at least~$2$ arguments and impose the
condition that in any parse tree, neither of them are permitted
to have themselves as children. Also, by the commutativity of the
sum operator, we impose the condition that the summands of each
sum appear in the following order:
First come all summands which are terminal symbols, then all summands 
which are concatenations, and finally all starred summands.
Also, we can safely omit all 
subexpressions of the form $s^{**}$, $s^{*}+\epsilon$, 
$(s+\epsilon)^{*}$, $s+\epsilon+\epsilon$: 
occurrences of these can be
replaced with occurrences of $s^*$, $s^{*}$, 
$s^{*}$, and $s+\epsilon$, respectively. 
Here the latter subexpressions have
size no larger than the former ones, and  
this holds for all size measures considered.
These observations immediately lend themselves for 
a simple unambiguous grammar, such as the one 
listed in Table~\ref{tab:grammar-simple}. 
The meaning of the variables is as follows:

\begin{table}
{
\renewcommand{\arraystretch}{1.1}
\renewcommand{\tabcolsep}{2mm}
\centerline{   
\begin{tabular}{|rl|}
\hline
\hline
$S \to $ & $\,  Q \mid A \mid T \mid C \mid K  $\\
\hline
\hline
$Q \to $ & $\, A \mathit{{}+ \epsilon} \mid 
          T \mathit{{}+ \epsilon} \mid 
          C\mathit{{}+ \epsilon}  $\\
\hline
\hline
$A \to$ & $\,  T + A_T \mid C + A_C \mid K + A_K $\\
$A_T \to $ & $\, T \mid T + A_T \mid A_C $\\
$A_C \to $ & $\, C \mid C + A_C \mid A_K$\\
$A_K \to $ & $\, K \mid K + A_K$\\
\hline
\hline
$T    \to $ & $\, a_1\mid a_2\mid \cdots \mid a_k$\\
\hline
\hline
$C \to $ & $\, C_0\, C_0 \mid C_0\, C $\\
$C_0 \to$ & $\, (Q) \mid (A) \mid  T \mid K$\\
\hline
\hline
$K \to $ & $\, (A)\,* \mid T\,* \mid (C)\,* $\\
\hline
\hline
\end{tabular}
}
}
\caption{A simple unambiguous grammar for generating at least one
shortest regular expression for each regular language.}\label{tab:grammar-simple}
\end{table}

\begin{itemize}
\item[$S$] generates all regular expressions obeying the 
abovementioned format. Among them,
\item[$Q$] generates those expressions of the form~$r+\epsilon$,
\item[$A$] generates those of the form $r+s$, i.e. ``additions'',
\item[$T$] generates those which are terminal symbols,
\item[$C$] generates those of the form $rs$, i.e. concatenations,
\item[$C_0$] generates the ``factors'' apppearing inside 
concatenations (which are themselves not concatenations), and
\item[$K$] generates those of the form $r^*$, i.e. Kleene stars; 
\end{itemize}
finally, the ``summands'' in expressions of type~$A$ are subdivided 
into subtypes~$A_T$, $A_C$ and $A_K$, used for handling summands 
which are terminal symbols, concatenations, or Kleene stars, respectively.

In the special case of unary alphabets, not only union, but also 
concatenation (again viewed as a variadic operator) is 
commutative. In this case, we may impose a similar ordering of 
factors as done for summands, and thus we can replace the rule 
with~$C$ as left-hand side with the rules given in 
Table~\ref{tab:grammar-simple2}.
\begin{table}
{
\renewcommand{\arraystretch}{1.1}
\renewcommand{\tabcolsep}{2mm}
\centerline{   
\begin{tabular}{|rl|}
\hline
$C \to$ &  $(Q)C_Q \mid (A)C_A \mid TC_T \mid KC_K$\\
$C_Q \to $ & $(Q) \mid (Q)C_Q \mid C_A$\\
$C_A \to $ & $A \mid (A)C_A \mid C_T$\\
$C_T \to $ & $T \mid TC_T \mid C_K$\\
$C_K \to $ & $K \mid KC_K $\\
\hline
\end{tabular}
}
}
\caption{Rules for concatenation over unary alphabets, which in 
that case is commutative.}\label{tab:grammar-simple2}
\end{table}

\subsection{A grammar based on strong star normal form}

We now refine the above approach by considering only regular 
expressions in strong star normal form~\cite{Gruber&Gulan:2010}, 
a notion that we recall in the following.
 
%compare Definition~\ref{defn:senf} in the chapter on descriptional 
%complexity of regular languages (Chapter~\ref{chapterDCRS}).
Since~$\emptyset$ is only needed to denote the empty set, and the need
for~$\varepsilon$ can be substituted by the operator $L^? = L \cup
\{\epsilon\}$, an alternative syntax introduces also the $^?$-operator
and instead forbids the use of~$\emptyset$ and~$\varepsilon$ inside
non-atomic expressions. The definition of strong star normal form is
most conveniently given for this alternative syntax.

\begin{definition}\label{defn:senf}
The operators $\circ$ and $\bullet$ are defined 
on regular expressions. The first operator is given by: 
$a^\circ = a$, for $a\in \Sigma$; 
$(r+s)^\circ = r^\circ + s^\circ$; $r^{?\circ} = r^\circ$; 
$r^{*\circ}  =r^\circ$; finally, 
$(rs)^\circ= rs$, if $\epsilon \notin L(rs)$ and
$r^\circ +s^\circ$ otherwise.
The second operator is given by: 
$a^\bullet = a$, for $a\in \Sigma$; 
$(r+s)^\bullet = r^\bullet + s^\bullet$; $(rs)^\bullet = r^\bullet s^\bullet$;
$r^{*\bullet}  = r^{\bullet\circ*}$; finally, $r^{?\bullet} = r^{\bullet}$, if 
$\epsilon\in L(r)$ and $r^{?\bullet} = r^{\bullet?}$ otherwise. 
The {\em strong star normal form} of an expression~$r$ 
is then defined as $r^\bullet$.
\end{definition}

An easy induction shows that the transformation into 
strong star normal form preserves the described 
language, and that it is weakly monotone 
with respect to all usual size measures. 
We sketch a proof for the case of ordinary length. 

\begin{lemma}
Let $r$ be a regular expression without 
occurrences of the symbol~$\emptyset$, 
and let $r^\bullet$ be its strong star normal form.  
Then $\ord(r^\bullet) \le \ord(r)$.  
\end{lemma}
\begin{proof}[Proof Sketch.]
First of all, we may safely assume that~$r$ does 
not contain any subexpressions ruled out by the 
grammar of the previous section, such 
as~$\epsilon+\epsilon$; the transformation into 
strong star normal form subsumes 
these reductions anyway. 

Recall the definition of the auxiliary operator~$^\circ$ 
in the definition of strong star normal form 
%(Chapter~\ref{chapterDCRS}, 
(Definition~\ref{defn:senf}).
The proof relies on the following claim:
If $\epsilon \in L(r)$ and $L(r)\neq \{\epsilon\}$, 
then $\ord(r^\circ)\le \ord(r)-1$;
otherwise, $\ord(r^\circ)\le \ord(r)$. This claim can be 
proved by induction while excluding the cases 
$L(r) = \emptyset, \{\epsilon\}$. 
The base cases are easy; the induction step is most 
interesting in the case $r = st$. If $\epsilon \notin L(st)$, 
then $r^\circ = st$ and the claim holds; otherwise 
$r^\circ = s^\circ+t^\circ$ with $\epsilon \in L(s)$ 
and $\epsilon\in L(t)$. We can apply the induction 
hypothesis twice to deduce 
$\ord(s^\circ) + \ord(t^\circ) \le \ord(s) + \ord(t)- 2$,
and thus $\ord(s^\circ + t^\circ) \le \ord(st)-1$, as 
desired. Notice that, as union has lower precedence than
concatenation, this step never introduces new parentheses.
The induction step in the other cases is even easier. 
\end{proof}

Since every regular language is represented by at least 
one shortest regular expression
in strong normal form (with respect to all three 
considered size measures), it suffices to enumerate 
those expressions in normal form. 
Our improved grammar will be based on the following 
simple observation on expressions in strong star normal form: 

\begin{lemma}
If $s^*$ or $s + \epsilon$ appears as a subexpression of an expression in star 
normal form, then $\epsilon \notin L(s)$. \qed
\end{lemma}

To exploit this fact, for each subexpression
we need to keep track of whether it denotes the empty word. 
This can of course 
be done with dynamic programming, by using rules such as 
$\epsilon \in L(rs)$ iff $\epsilon\in L(r)$ and $\epsilon \in L(s)$.
Since in addition every subexpression either denotes the empty word or not,
it is easy to extend the above grammar to incorporate these rules
while retaining the property of being unambiguous.

\begin{table} 
{
\renewcommand{\arraystretch}{1.1}
\renewcommand{\tabcolsep}{2mm}  
\begin{tabular}{|rlrl|}
\hline
\hline
\multicolumn{2}{|r}{$S \to$} & \multicolumn{2}{l|}{$\, S^{+} \mid S^{-} $}\\
\hline
\hline
$S^{+} \to$ & $\, Q^+  \mid A^{+} \mid C^{+} \mid K^+$
&$S^{-} \to $ & $\,A^{-} \mid  T^{-} \mid C^{-}$\\
\hline\hline
$Q^+ \to $ & $\, A^{-}\mathit{{}+\epsilon}  \mid
T^{-}\mathit{{}+\epsilon}\mid C^{-}\mathit{{}+\epsilon}$&&\\
\hline\hline
$A^{+}  \to$ & $\, T^{-} + A^{+}_C  \mid C^{-} + A^{+}_C \mid$
  &\multirow{3}{*}{$A^{-}   \to $} & \multirow{3}{*}{$\, T^{-} + A^{-}_T \mid C^{-} + A^{-}_C $}\\
& $A^{-}+A^{+}_C \mid C^{+} + A^{+}_C \mid$ && \\
& $K^{+} + A^{+}_K$&&\\
&&$A^{-}_T  \to $ & $\, T^{-}  \mid T^{-} + A^{-}_T \mid A_C^{-}$\\
$A^{+}_C \to $ & $\, C^{+} \mid C^{+} + A^{+}_C\mid A^{+}_K$
&$A^{-}_C \to $ & $\, C^{-} \mid C^{-} + A^{-}_C$\\
$A^{+}_K \to$ & $\, K^+ \mid K^+ + A^{+}_K$&&\\
\hline\hline
&&$T^{-}    \to $ & $\, a_1\mid a_2\mid \cdots \mid a_k$\\
\hline\hline
\multirow{2}{*}{$C^{+} \to $} & \multirow{2}{*}{$\, C^{+}_0 C^{+}_0 \mid C^{+}_0 C^{+} $}
&$C^{-} \to $ & $\,   C_0^-\,C_0^- \mid C_0^-\,C_0^+ \mid C_0^+\,C_0^-\mid$\\
&& & $\,  C_0^-\,C^-   \mid C_0^-\,C^+  \mid C_0^+\, C^-$\\
$C^{+}_0 \to$ & $(Q^+) \mid \,(A^{+}) \mid K^+ $
&$C^{-}_0 \to $ & $\, (A^{-}) \mid T^{-}$\\
\hline\hline
$K^+ \to $ & $\, (A^{-})\,* \mid T^{-}\,*\mid (C^{-})\,* $&&\\
\hline
\hline
%$ & $\, (C^{-}_0\mid C^{+}_0) {}\bull{} C^{-}_0 \mid (C^{-}_0\mid C^{+}_0) {}\bull{} C^{-} $\\
%        $ & $\mid C^{-}_0 {}\bull{} C^{+}_0 \mid C^{-}_0 {}\bull{} C^{+}$\\
\end{tabular}
}
\caption{A better unambiguous grammar generating at least one
shortest regular expression (in strong star normal form)
for each regular language.}\label{gram:improved}
\end{table}

Notice that most variables now come in  
an $\epsilon$-flavor (for example, the variable $A^{+}$) 
and in an $\epsilon$-free flavor (for example, the variable $A^{-}$). 
Moreover, the summands inside sums appear in the following order,
which is a refinement of the summand ordering devised previously:
First come all summands which are terminal symbols, then all summands 
which are $\epsilon$-free concatenations, then all concatenations
with~$\epsilon$ in the denoted language, and finally all starred
summands. To illustrate this ordering, we give the most important 
steps of the unique derivation for the 
expression~$a_1 + a_2a_3 + (a_4+\epsilon)(a_5+\epsilon) + a_6^{~*}$:

\begin{align*}
S & \derivestar A^{-} + A^{+}_C \derives 
               T^{-} + A^-_T + A^+_C \derives a_1 + A^-_T +A^+_C \\
& \derives a_1 +  A^-_C +A^+_C \derives a_1 +  C^- + A^+_C 
\derivestar a_1 + a_2a_3 + A^+_C\\
& \derives a_1 + a_2a_3 + C^+ + A^+_C \derivestar  
a_1 + a_2a_3 + (a_4+\epsilon)(a_5+\epsilon) + A^+_C \\
& \derives a_1 + a_2a_3 + (a_4+\epsilon)(a_5+\epsilon) + A^+_K
\derives a_1 + a_2a_3 + (a_4+\epsilon)(a_5+\epsilon) + K^+\\
& \derivestar a_1 + a_2a_3 + (a_4+\epsilon)(a_5+\epsilon) + a_6^{~*}
\end{align*}

The following proposition, giving the correctness of the 
improved grammar, can be proved by induction on the  
minimum required regular expression size. 
Table~\ref{tab:upper} lists the upper bounds obtained
through this grammar.\footnote{The Maple worksheets used to 
derive these bounds can be accessed at the second author's 
personal homepage via 
\url{http://math.stanford.edu/~jlee/automata/}} 

\begin{proposition}
The grammar in Table~\ref{gram:improved} is unambiguous and,
for each regular language, generates at least one regular
expression of minimal ordinary length 
(respectively: reverse polish length, alphabetic width)
representing it. \qed
\end{proposition} 
\begin{table}[htpb]
\setlength{\tabcolsep}{0.125in}
\center
\begin{tabular}{|r|lll|}
\hline
\hline
$k$ & ordinary & reverse polish & alphabetic \\ \hline
$1$ & $O(2.5946^n)$ & $O(2.7422^n)$ & \multirow{6}{*}{
$\left.\begin{aligned}
\\
\\
\\
\\
\end{aligned}\right\}$
$O\left(k^n\cdot 21.5908^n\right)$ } \\
$2$ & $O(4.2877^n)$ & $O(3.9870^n)$ & \\
$3$ & $O(5.4659^n)$ & $O(4.7229^n)$ &  \\
$4$ & $O(6.5918^n)$ & $O(5.3384^n)$ &  \\
$5$ & $O(7.6870^n)$ & $O(5.8780^n)$ & \\
$6$ & $O(8.7624^n)$ & $O(6.3643^n)$ & \\
\hline
\hline
\end{tabular}
\caption{Summary of upper bounds on $R_k(n)$ for $k=1,2,\ldots,6$ and 
various size measures. For ordinary length, we used the simple 
grammar in Table~\ref{tab:grammar-simple}, because the computation  
for the improved grammar ran out of computational resources.
For reverse polish length, we used the simple grammar 
for bootstrapping the bounds.
}
\label{tab:upper}
\end{table}

\begin{table}[htpb]
\setlength{\tabcolsep}{0.125in}
\center
\begin{tabular}{|r|lll|}
\hline
\hline
$k$ & ordinary & reverse polish & alphabetic \\ \hline
$1$ & $O(2.1793^n)$ & $O(2.0795^n)$ & $O(10.9822^n)$\\
$2$ & $O(3.8145^n)$ & $O(3.3494^n)$ & \multirow{5}{*}{
$\left.\begin{aligned}
\\
\\
\\
\\
\end{aligned}\right\}$
$O\left(k^n\cdot 12.2253^n\right)$ } \\
$3$ & $O(4.9019^n)$ & $O(4.0315^n)$ &  \\
$4$ & $O(5.8234^n)$ & $O(4.6121^n)$ &  \\
$5$ & $O(6.8933^n)$ & $O(5.1268^n)$ &\\
$6$ & $O(7.9492^n)$ & $O(5.5939^n)$ &\\
\hline
\hline
\end{tabular}
%k=2,3,4,rpn: bootstrapped from the general case in Table~\ref{tab:upper}.
%k=2, ord: bootstrapped from the case k=3
\caption{Summary of upper bounds for $k=1,2,..,6$ and 
various size measures in the case of finite languages. 
For reverse polish length, we bootstrapped from the values in Table~\ref{tab:upper};
for ordinary length, we bootstrapped the case $k=2$ from the upper 
bound obtained for $k=3$.
}
\label{tab:upper-finite}
\end{table}

%% file: exact.tex
Tables~\ref{sf_ord} to \ref{gen_alph} give exact
numbers for the number of regular languages
representable by a regular expression of size~$n$,
but not by any of size less than~$n$.

We explain how these numbers were obtained.\footnote{The C++ source code 
of the software used to compute these numbers 
can be accessed at the second author's personal homepage via 
\url{http://math.stanford.edu/~jlee/automata/}} 
Using the upper bound grammars described previously,
a dynamic programming approach was taken to
produce (in order of increasing regular expression
size) the regular expressions generated by each
non-terminal.
To account for duplicates, each regular expression
was transformed into a DFA, minimized and relabelled
via a breadth-first search to produce a canonical
representation.
Using these representations as hashes, any regular
expression matching a previous one generated by
the same non-terminal was simply ignored.

%%%%%
\begin{table}[htpb]
\begin{minipage}[b]{0.48\textwidth}
\center
\setlength{\tabcolsep}{0.05in}
\begin{tabular}{r|rrrr}
$k$ & 1 & 2 & 3 & 4 \\ \hline
1 & 3 & 4 & 5 & 6 \\
2 & 1 & 4 & 9 & 16 \\
3 & 2 & 11 & 33 & 74 \\
4 & 3 & 28 & 117 & 336 \\
5 & 3 & 63 & 391 & 1474 \\
6 & 5 & 156 & 1350 & 6560 \\
7 & 5 & 358 & 4546 & 28861 \\
8 & 8 & 888 & 15753 & 128720 \\
9 & 9 & 2194 & 55053 & 578033 \\
10 & 14 & 5665 & 196185 & 2624460 \\%NEW! k=4,n=10
\end{tabular}
\caption{\label{sf_ord} Ordinary length, finite languages}
\end{minipage}
\begin{minipage}[b]{0.48\textwidth}
\center
\setlength{\tabcolsep}{0.05in}
\begin{tabular}{r|rrrr}
$k$ & 1 & 2 & 3 & 4 \\ \hline
1 & 3 & 4 & 5 & 6 \\
2 & 2 & 6 & 12 & 20 \\
3 & 3 & 17 & 48 & 102 \\
4 & 4 & 48 & 192 & 520 \\
5 & 5 & 134 & 760 & 2628 \\
6 & 9 & 397 & 3090 & 13482 \\
7 & 12 & 1151 & 12442 & 68747 \\
8 & 17 & 3442 & 51044 & 354500 \\
9 & 25 & 10527 & 211812 & 1840433 \\%NEW! k=4,n=9: 1840433
10 & 33 & 32731 & 891228 & 
\end{tabular}
\caption{\label{gen_ord} Ordinary length, general case}
\end{minipage}
\end{table}
\begin{table}[htpb]
\begin{minipage}[b]{0.48\textwidth}
\center
\setlength{\tabcolsep}{0.05in}
\begin{tabular}{r|rrrr}
$k$ & 1 & 2 & 3 & 4 \\ \hline
1 & 3 & 4 & 5 & 6 \\
3 & 2 & 7 & 15 & 26 \\
5 & 3 & 25 & 85 & 202 \\
7 & 5 & 109 & 589 & 1917 \\
9 & 9 & 514 & 4512 & 20251 \\
11 & 14 & 2641 & 37477 & 231152 \\%NEW! k=4,n=11: 231152
13 & 24 & 14354 & 328718 & 2780936 \\%NEW! k=4, n=13
15 & 41 & 81325 & 2998039 \\ %NEW! k=3,n=15 there was a wrong entry: 231152 most probably not k=3,n=13
17 & 71 & 475936 \\
19 & 118 & 2854145 %NEW!k=2,n=19
\end{tabular}
\caption{\label{sf_rpn} Reverse polish length, finite languages}
\end{minipage}
\begin{minipage}[b]{0.48\textwidth}
\center
\setlength{\tabcolsep}{0.05in}
\begin{tabular}{r|rrrr}
$k$ & 1 & 2 & 3 & 4 \\ \hline
1 & 3 & 4 & 5 & 6 \\
2 & 1 & 2 & 3 & 4 \\
3 & 2 & 7 & 15 & 26 \\
4 & 2 & 13 & 33 & 62 \\
5 & 3 & 32 & 106 & 244 \\
6 & 4 & 90 & 361 & 920 \\
7 & 6 & 189 & 1012 & 3133 \\
8 & 7 & 580 & 3859 & 13529 \\
9 & 11 & 1347 & 11655 & 48388 \\
10 & 15 & 3978 & 43431 & 208634
\end{tabular}
\caption{\label{gen_rpn} Reverse polish length, general case}
\end{minipage}
\vfill
\end{table}
\begin{table}[htpb]
\begin{minipage}[b]{0.48\textwidth}
\center
\setlength{\tabcolsep}{0.05in}
\begin{tabular}{r|rrrr}
$k$ & 1 & 2 & 3 & 4 \\ \hline
0 & 2 & 2 & 2 & 2 \\
1 & 2 & 4 & 6 & 8 \\
2 & 4 & 24 & 60 & 112 \\
3 & 8 & 182 & 806 & 2164 \\
4 & 16 & 1652 & 13182 & 51008 \\
5 & 32 & 16854 & 242070 & 1346924 \\
6 & 64 & 186114 & 4785115 %NEW!k=3,n=6
\end{tabular}
\caption{\label{sf_alph} Alphabetic width, finite languages}
\end{minipage}
\begin{minipage}[b]{0.48\textwidth}
\center
\setlength{\tabcolsep}{0.05in}
\begin{tabular}{r|rrrr}
$k$ & 1 & 2 & 3 & 4 \\ \hline
0 & 2 & 2 & 2 & 2 \\
1 & 3 & 6 & 9 & 12 \\
2 & 6 & 56 & 150 & 288 \\
3 & 14 & 612 & 3232 & 9312 \\
4 & 30 & 7923 & 82614 & 357911 \\
5 & 72 & 114554 & 2332374 \\%NEW!k=3,n=5
6 & 155 & 1768133 \\
%7 & 343 \\
%8 & 731 \\
%9 & 1600 \\
%10 & 3407
\end{tabular}
\caption{\label{gen_alph} Alphabetic width, general case}
\end{minipage}
\vfill
\end{table}
%%%%%